\documentclass[numbook]{article}
\usepackage{fullpage}
\usepackage{setspace}
\usepackage{graphicx}
\usepackage{delarray}
\usepackage{bbm}
\usepackage{amssymb}
\usepackage{amsmath}
\usepackage{amsthm}
\usepackage{algorithm}
\usepackage{algorithmic}
\usepackage{bm}
\usepackage[T1]{fontenc}
\usepackage{xr}

\newcommand{\argmax}{\operatornamewithlimits{argmax}} 

\newtheoremstyle{dotless}{}{}{\itshape}{}{\bfseries}{}{ }{}
\theoremstyle{dotless}

\newtheorem{corollary}{Corollary}[section]

\newtheorem{theorem}{Theorem}[section]

\newtheorem{lemma}{Lemma}[section]

\newcommand{\R}{\mathbb{R}}
\newcommand{\E}{\mathbb{E}}
\newcommand{\PP}{\mathbb{P}}
\newcommand{\I}{\mathbb{I}}

\newcommand{\bs}{\boldsymbol}

\newcommand{\be}{\begin{equation}}
\newcommand{\ee}{\end{equation}}
\newcommand{\bea}{\begin{eqnarray}}
\newcommand{\eea}{\end{eqnarray}}
\newcommand{\bfl}{\begin{flalign}}
\newcommand{\efl}{\end{flalign}}
\newcommand{\bfc}{\begin{figure}\begin{center}}
\newcommand{\efc}{\end{center}\end{figure}}
\newcommand{\n}{\nonumber}
\newcommand{\dd}{\mathrm{d}}
\newcommand{\nin}{\noindent}
\newcommand{\mc}{\mathcal}
\newcommand{\ds}{\displaystyle}

\newcommand{\lan}{\langle}
\newcommand{\ran}{\rangle}
\renewcommand{\bar}{\overline}

\title{Average-Case Performance of Rollout Algorithms \\ for Knapsack Problems\thanks{Supported by NSF grant 1029603. The first author is supported in part by a NSF graduate research fellowship.}}

\author{%
Andrew Mastin\thanks{Department of Electrical Engineering and Computer Science, Laboratory for Information and Decision Systems, Massachusetts Institute of Technology, Cambridge, MA 02139. Corresponding author.  \texttt{mastin@mit.edu}}%
\and
Patrick Jaillet\thanks{Department of Electrical Engineering and Computer Science
, Laboratory for Information and Decision Systems,
Operations Research Center, Massachusetts Institute of Technology, Cambridge, MA 02139; \texttt{jaillet@mit.edu}}%
}

\date{}

\begin{document}

\maketitle

\begin{abstract}
Rollout algorithms have demonstrated excellent performance on a variety of dynamic and discrete optimization problems. Interpreted as an approximate dynamic programming algorithm, a rollout algorithm estimates the value-to-go at each decision stage by simulating future events while following a greedy policy, referred to as the base policy. While in many cases rollout algorithms are guaranteed to perform as well as their base policies, there have been few theoretical results showing additional improvement in performance. In this paper we perform a probabilistic analysis of the subset sum problem and knapsack problem, giving theoretical evidence that rollout algorithms perform strictly better than their base policies. Using a stochastic model from the existing literature, we analyze two rollout methods that we refer to as the \textit{consecutive rollout} and \textit{exhaustive rollout}, both of which employ a simple greedy base policy. For the subset sum problem, we prove that after only a single iteration of the rollout algorithm, both methods yield at least a $30\%$ reduction in the expected gap between the solution value and capacity, relative to the base policy. Analogous results are shown for the knapsack problem.
\end{abstract}


\clearpage

\section{Introduction}
Rollout algorithms provide a natural and easily implemented approach for approximately solving many discrete and dynamic optimization problems. Their motivation comes from problems that can be solved using classical dynamic programming, but for which determining the value function (or value-to-go function) is computationally infeasible. The rollout technique estimates these values by simulating future events while following a simple greedy/heuristic policy, referred to as the base policy. In most cases the rollout algorithm is ensured to perform as well as its base policy \cite{bertsekas97}. As shown by many computational studies, the performance is often much better than the base policy, and sometimes near optimal \cite{bertsekas99}.


Theoretical results showing a strict improvement of rollout algorithms over base policies have been limited to average-case asymptotic bounds on the breakthrough problem and a worst-case analysis of the knapsack problem \cite{bertsekas07,bertazzi12}. The latter work motivates a complementary study of rollout algorithms for knapsack-type problems from an average-case perspective, which we provide in this paper. Our goals are to give theoretical evidence for the utility of rollout algorithms and to contribute to the knowledge of problem types and features that make rollout algorithms work well. We anticipate that our proof techniques may be helpful in achieving performance guarantees on similar problems.


We use a stochastic model directly from the literature that has been used to study a wide variety of greedy algorithms for the subset sum problem \cite{bt91}. This model is extended in a natural manner for our analysis of the knapsack problem. We analyze two rollout techniques that we refer to as the \textit{consecutive rollout} and the \textit{exhaustive rollout}, both of which use the same base policy. The first algorithm sequentially processes the items and at each iteration decides if the current item should be added to the knapsack. During each iteration of the exhaustive rollout, the algorithm decides which one of the available items should be added to the knapsack. The base policy is a simple greedy algorithm that adds items until an infeasible item is encountered.

For both techniques, we derive bounds showing that the expected performance of the rollout algorithms is strictly better than the performance obtained by only using the base policy. For the subset sum problem, this is demonstrated by measuring the gap between the total value of packed items and capacity. For the knapsack problem, the difference between total profits of the rollout algorithm and base policy is measured. The bounds are valid after only a single iteration of the rollout algorithm and hold for additional iterations.

The organization of the paper is as follows. In the remainder of this section we review related work, and we introduce our notation in Section \ref{notation}. Section \ref{modelgreedy} describes the stochastic models in detail and derives important properties of the \textit{blind greedy} algorithm, which is the algorithm that we use for a base policy. Results for the consecutive rollout and exhaustive rollout are shown in Section \ref{consecutive} and Section \ref{exhaustive}, respectively; these sections contain the most important proofs used in our analysis. A conclusion is given in Section \ref{conclusion}. A list of symbols, omitted proofs, and an appendix with evaluations of integrals are provided in the supplementary material.

\subsection{Related work}
Rollout algorithms were introduced by Tesauro and Galperin as online Monte-Carlo search techniques for computer backgammon \cite{tesauro96}. The application to combinatorial optimization was formalized by Bertsekas, Tsitsiklis, and Wu \cite{bertsekas97}. They gave conditions under which the rollout algorithm is guaranteed to perform as well as its base policy, namely if the algorithm is \textit{sequentially consistent} or \textit{sequentially improving}, and presented computational results on a two-stage maintenance and repair problem. The application of rollout algorithms to approximate stochastic dynamic programs was provided by Bertsekas and Casta\~{n}on, where they showed extensive computational results on variations of the quiz problem \cite{bertsekas99}.  Rollout algorithms have since shown strong computational results on a variety of problems including vehicle routing, fault detection, and sensor scheduling \cite{secomandi01, tu02, li09}.

Beyond simple bounds derived from base policies, the only theoretical results given explicitly for rollout algorithms are average-case results for the breakthrough problem and worst-case results for the 0-1 knapsack problem \cite{bertazzi12,bertsekas07}. In the breakthrough problem, the objective is to find a valid path through a directed binary tree where some edges are blocked. If the free (non-blocked) edges occur with probability $p$, independent of other edges, a rollout algorithm has a $O(N)$ larger probability of finding a free path in comparison to a greedy algorithm \cite{bertsekas07}.  Performance bounds for the knapsack problem were recently shown by Bertazzi \cite{bertazzi12}, who analyzed the rollout approach with variations of the decreasing density greedy (DDG) algorithm as a base policy. The DDG algorithm takes the best of two solutions: the one obtained by adding items in order of non-increasing profit to weight ratio, as long as they fit, and the solution resulting from adding only the item with highest profit.  He demonstrated that from a worst-case perspective, running the first iteration of a rollout algorithm (specifically, what we will refer to as the exhaustive rollout algorithm) improves the approximation guarantee from $\frac{1}{2}$ (bound provided by the base policy) to $\frac{2}{3}$.

An early probabilistic analysis of the subset sum problem was given by d'Atri and Puech \cite{d'atri82}. Using a discrete version of the model used in our paper, they analyzed the expected performance of greedy algorithms with and without sorting. They showed an exact probability distribution for the gap remaining after the algorithms and gave asymptotic expressions for the probability of obtaining a non-zero gap. These results were refined by Pferschy, who gave precise bounds on expected gap values for greedy algorithms \cite{pferschy97}.

A very extensive analysis of greedy algorithms for the subset sum problem was given by Borgwardt and Tremel \cite{bt91}. They introduced the continuous model that we use in this paper and derived probability distributions of gaps for a variety of greedy algorithms. In particular, they showed performance bounds for a variety of prolongations of a greedy algorithm, where a different strategy is used on the remaining items after the greedy policy is no longer feasible. They also analyzed cases where items are ordered by size prior to use of the greedy algorithms.

In the area of probabilistic knapsack problems, Szkatula and Libura investigated the behavior of greedy algorithms, similar to the blind greedy algorithm used in our paper, for the knapsack problem with fixed capacity. They found recurrence equations describing the weight of the knapsack after each iteration and solved the equations for the case of uniform weights \cite{szkatula83}. In later work they studied asymptotic properties of greedy algorithms, including conditions for the knapsack to be filled almost surely as $n \rightarrow \infty$ \cite{szkatula87}. 

There has been some work on asymptotic properties of the decreasing density greedy algorithm for probabilistic knapsack problems. Diubin and Korbut showed properties of the asymptotical tolerance of the algorithm, which characterizes the deviation of the solution from the optimal value \cite{diubin03}. Similarly, Calvin and Leung showed convergence in distribution between the value obtained by the DDG algorithm and the value of the knapsack linear relaxation \cite{calvin03}.

\section{Notation}
\label{notation}
Before we describe the model and algorithms, we summarize our notation. Since we must keep track of ordering in our analysis, we use sequences in place of sets and slightly abuse notation to perform set operations on sequences. These operations will mainly involve index sequences, and our index sequences will always contain unique elements. Sequences will be denoted by bold letters. If we wish for $\bs{S}$ to be the increasing sequence of integers ranging from $2$ to $5$, we write $\bs{S} = \langle 2,3,4,5 \rangle$. We then have $2 \in \bs{S}$ while $1 \notin \bs{S}$. We also say that $\langle 2,5 \rangle \subseteq \bs{S}$ and  $\bs{S} \setminus \langle 3 \rangle = \langle 2,4,5\rangle$. The concatenation of sequence $\bs{S}$ with sequence $\bs{R}$ is denoted by $\bs{S}:\bs{R}$. If $\bs{R} = \langle 1,7 \rangle$ then $\bs{S}:\bs{R} = \langle 2,3,4,5,1,7 \rangle$. A sequence is indexed by an index sequence if the index sequence is shown in the subscript. Thus $\bs{a_S}$ indicates the sequence $\langle a_2, a_3, a_4, a_5 \rangle$. For a sequence to satisfy equality with another sequence, equality must be satisfied element by element, according to the order of the sequence. We use the notation $\bs{S^i}$ to denote the sequence $\bs{S}$ with item $i$ moved to the front of the sequence: $\bs{S^3} = \langle 3,2,4,5 \rangle$.

The notation $\PP(\cdot)$ indicates probability and $\E[\cdot]$ indicates expectation. We define $\bar{\E}[\cdot] := 1- \E[\cdot]$. For random variables, we will use capital letters to denote the random variable (or sequence) and lowercase letters to denote specific instances of the random variable (or sequence). The probability density function for a random variable $X$ is denoted by $f_X(x)$. For random variables $X$ and $Y$, we use $f_{X|Y}(x|y)$ to denote the conditional density of $X$ given the event $Y=y$. When multiple variables are involved, all variables on the left side of the vertical bar are conditioned on all variables on the right side of vertical bar. The expression $f_{X,Y|Z,W}(x,y|z,w)$ should be interpreted as $f_{(X,Y) | (Z,W)}((x,y)|(z,w))$ and not $f_{X,(Y|Z),W}(x,(y|z),w)$, for example. Events are denoted by the calligraphic font, such as $\mc{A}$, and the disjunction of two events is shown by the symbol $\vee$. We often write conditional probabilities of the form $\PP(\cdot|X=x,Y=y, \mc{A})$ as $\PP(\cdot | x,y, \mc{A})$. The notation $\mc{U}[a,b]$ indicates the density of a uniform random variable on interval $[a,b]$. The indicator function is denoted by $\I(\cdot)$ and the positive part of an expression is denoted by $(\cdot)_+$. Finally, we use the standard symbols for assignment ($\leftarrow$), definition ($:=$), the positive real numbers ($\R^+$), and asymptotic growth ($O(\cdot)$).


\section{Stochastic model and blind greedy algorithm}
\label{modelgreedy}
In the knapsack problem, we are given a sequence of items $\bs{I} = \langle 1,2,\ldots,n \rangle $ where each item $i \in \bs{I}$ has a weight $w_i \in \R^+$ and profit $p_i \in \R^+$. Given a knapsack with capacity $b \in \R^+$, the goal is to select a subset of items with maximum total profit such that the total weight does not exceed the capacity. This is given by the following integer linear program.
\be
\begin{array}{cll}
$max$ & \ds \sum_{i =1}^n p_i x_{i} & \\
$s.t.$ & \ds \sum_{i =1}^n w_{i} x_{i} \le b & \\
& x_{i} \in \{0,1\} & i = 1,\ldots,n. \\
\end{array}
\ee
For the subset sum problem, we simply have $p_i = w_i$ for all $i \in \bs{I}$. 

We use the stochastic subset sum model given by Borgwardt and Tremel \cite{bt91}, and a variation of this model for the knapsack problem. In their subset sum model, for a specified number of items $n$, item weights $W_i$ and the capacity $B$ are drawn independently from the following distributions:
\bea
W_i &\sim& \mc{U}[0,1],~~~i=1,\ldots,n, \n \\
B &\sim& \mc{U}[0,n].
\eea
Our stochastic knapsack model simply assigns item profits that are independently and uniformly distributed,
\bea
P_i &\sim& \mc{U}[0,1],~~~i=1,\ldots,n.
\eea
These values are also independent with respect to the weights and capacity.

For evaluating performance, we only consider cases where $\sum_{i=1}^n W_i > B$. In all other cases, any algorithm that tries adding all items is optimal. Since it is difficult to understand the stochastic nature of optimal solutions, we use $\E[B-\sum_{i \in \bs{S}} W_i|\sum_{i =1}^nW_i>B]$ as a performance metric for the subset sum problem, where $\bs{S}$ is the sequence of items selected by the algorithm of interest. This is the same metric used in \cite{bt91}, where they note with a simple symmetry argument that for all values of $n$,
\be
\PP \left (\sum_{i = 1}^n W_i > B \right ) = \frac{1}{2}.
\ee
For the knapsack problem, we directly measure the difference between the rollout algorithm profit and the profit given by the base policy, which we refer to as the gain of the rollout algorithm.

\algsetup{indent=2em}
\begin{algorithm}
\caption{\textsc{Blind-Greedy}}
\label{blindgreedyAlg}
\begin{algorithmic}[1]
\REQUIRE Item weight sequence  $\bs{w_I}$ where $\bs{I} = \langle 1,\ldots,n \rangle$, capacity $b$.
\ENSURE Feasible solution sequence $\bs{S}$, value $U$.
\STATE Initialize solution sequence $\bs{S} \leftarrow \langle \rangle$, remaining capacity $\bar{b} \leftarrow b$, and value $U \leftarrow 0$.
\FOR{$i=1$ to $n$ (each item)}
\IF{$w_i \le \bar{b}$ (item weight does not exceed remaining capacity)}
\STATE Add item $i$ to solution sequence, $\bs{S} \leftarrow \bs{S} : \langle i \ran$.
\STATE Update remaining capacity $\bar{b} \leftarrow \bar{b} - w_i$, and value $U \leftarrow U+p_i$.
\ELSE
\STATE Stop and return $\bs{S}$, $U$.
\ENDIF
\ENDFOR
\STATE Return $\bs{S}$, $U$.
\end{algorithmic}
\end{algorithm}


For both the subset sum problem and the knapsack problem, we use the \textsc{Blind-Greedy} algorithm, shown in Algorithm \ref{blindgreedyAlg}, as a base policy. The algorithm simply adds items (without sorting) until it encounters an item that exceeds the remaining capacity, then stops. Throughout the paper, we will sometimes refer to \textsc{Blind-Greedy} simply as the greedy algorithm.

\textsc{Blind-Greedy} may seem inferior to a greedy algorithm that first sorts the items by weight or profit to weight ratio and then adds them in non-decreasing value. Surprisingly, for the subset sum problem, it was shown in \cite{bt91} that the algorithm that adds items in order of non-decreasing weight (referred to as \textsc{Greedy 1S}) performs equivalently to \textsc{Blind-Greedy}. Of course, we cannot say the same about the knapsack problem. A greedy algorithm that adds items in decreasing profit to weight ratio is likely to perform much better. Applying our analysis to a sorted greedy algorithm requires work beyond the scope of this paper.

In analyzing \textsc{Blind-Greedy}, we refer to the index of the first item that is infeasible as the \textit{critical item}. Let $K$ be the random variable for the index of the critical item, where $K=0$ indicates that there is no critical item (meaning $\sum_{i=1}^n W_i \le B$). Equivalently, assuming $\sum_{i=1}^nW_i > B$, the critical item index satisfies
\bea
\sum_{i=1}^{K-1}W_i \le B < \sum_{i=1}^{K}W_i. \label{critdef}
\eea
We will refer to items with indices $i <K$ as packed items. We then define the \textit{gap} of \textsc{Blind-Greedy} as
\be
\label{gdef}
G := B - \sum_{i=1}^{K-1} W_i,
\ee
for $K>0$. The gap is relevant to both the subset sum problem and the knapsack problem. For the knapsack problem, we define the gain of the rollout algorithm as 
\be
Z := \sum_{i \in \bs{R}} P_i - \sum_{i = 1}^{K-1} P_i,
\label{gaindef}
\ee
where $\bs{R}$ is the sequence of items selected by the rollout algorithm. A central result of \cite{bt91} is the following, which does not depend on the number of items $n$.
\begin{theorem}[Borgwardt and Tremel, 1991]
Independent of the critical item index $K>0$, the probability distribution of the gap obtained by \textsc{Blind-Greedy} satisfies
\be
\PP \left ( G \le g \left \vert \sum_{i=1}^nW_i > B \right. \right) = 2g-g^2,~~ 0 \le g \le 1,
\ee
\be
\E \left [ G  \left \vert \sum_{i=1}^nW_i > B \right. \right ] = \frac{1}{3}.
\ee
\label{greedythm}
\end{theorem}

Many studies measure performance using an approximation ratio (bounding the ratio of the value obtained by some algorithm to the optimal value) \cite{kellerer04, bertazzi12}. While this metric is generally not tractable under the stochastic model, we can observe a simple lower bound on the ratio of expectations of the value given by \textsc{Blind-Greedy} to the optimal value, for the subset sum problem\footnote{The expected ratio, rather than the ratio of expectations, may be a better benchmark here, but is less tractable}. A natural upper bound on the optimal solution is $B$, and the solution value given by \textsc{Blind-Greedy} is equal to $B-G$. Thus by Theorem \ref{greedythm} and linearity of expectations, the ratio of expected values is at least $\frac{\E[B-G]}{\E[B]} = 1-\frac{2}{3n}$. For $n \ge 2$, this value is at least $\frac{2}{3}$, which is the best worst-case approximation ratio derived in \cite{bertazzi12}. A similar comparison for the knapsack problem is not possible because there is no simple bound on the expected optimal solution value.

We describe some important properties of the \textsc{Blind-Greedy} solution that will be used in later sections and that provide a proof of Theorem \ref{greedythm}. For the proofs in this section as well as other sections, it is helpful to visualize the \textsc{Blind-Greedy} solution sequence on the nonnegative real line, as shown in Figure \ref{vis}.

\begin{figure}[h]
\begin{center}
\includegraphics[scale=0.7]{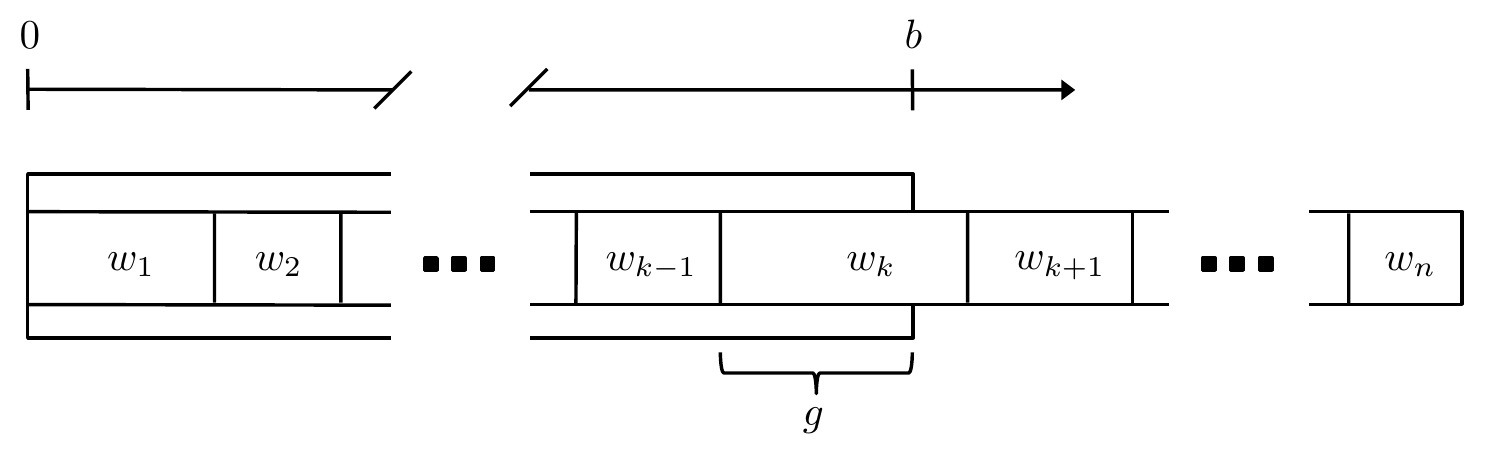}
\caption{Sequence given by \textsc{Blind-Greedy} on the nonnegative real line where $G=g$, $B=b$, and $\bs{W_S} = \bs{w_s}$. Each item $\ell = 1,\ldots,n$ occupies the interval $\left [ \sum_{i = 1}^{\ell-1}w_i, \sum_{i=1}^\ell w_i \right )$ and the knapsack is given on the interval $[0,b]$. The gap is the difference between the capacity and the total weight of the packed items. \label{vis}}
\end{center}
\end{figure}

Previous work on the stochastic model has demonstrated that the critical item index is uniformly distributed on $\{1,2,\ldots,n\}$ for cases of interest (i.e. $\sum_{i=1}^nW_i > B$) \cite{bt91}. In addition to this property, we show that the probability that a given item is critical is independent of weights of all other items\footnote{In other sections we follow the convention of associating the index $k$ with the random variable $K$. The index $\ell$ is used in this section to make the proofs clearer.}.
\begin{lemma}
\label{lemmakunif}
For each item $\ell =  1,\ldots,n$, for all subsequences of items $\bs{S} \subseteq \bs{I} \setminus \langle \ell \rangle$ and all weights $\bm{w_S}$, the probability that item $\ell$ is critical is
\be
\ds \PP(K=\ell | \bm{W_S} = \bm{w_S}) = \frac{1}{2n}.
\ee
\end{lemma}
\begin{proof}
Assume that we are given the weights of all items $\bm{W_I} = \bm{w_I}$. We can divide the interval $[0,n]$ into $n+1$ segments as a function of item weights as shown in Figure \ref{vis}, so that the $\ell$th segment occupies the interval $\left [ \sum_{i = 1}^{\ell-1}w_i, \sum_{i=1}^\ell w_i \right )$ for $\ell=1,\ldots,n$ and the last segment is on $\left [ \sum_{i = 1}^{n}w_i, n\right ]$. The probability that item $\ell$ is critical is the probability that $B$ intersects the $\ell$th segment.
Since $B$ is distributed uniformly over the interval $[0,n]$, we have
\be
\PP(K=\ell|\bm{W_I}=\bm{w_I}) = \frac{w_\ell}{n},
\ee
showing that this event only depends on $w_\ell$. Integrating over the uniform density of $w_\ell$ gives the result.
\end{proof}
An important property of this stochastic model, which is key for the rest of our development, is that conditioning on the critical item index only changes the weight distribution of the critical item; all other item weights remain independently distributed on $\mc{U}[0,1]$.
\begin{lemma}
\label{othersindep}
For any critical item $K>0$ and any subsequence of items $\bs{S} \subseteq \bs{I} \setminus \langle K \rangle$, the weights $\bs{W_S}$ are independently distributed on $\mathcal{U}[0,1]$, and $W_K$ independently follows the distribution
\bea
f_{W_K}(w_K) = 2w_K,~~ 0 \le w_K \le 1.
\eea
\end{lemma}
\begin{proof}
For any item $\ell = 1,\ldots,n$, consider the subsequence of items $\bs{S} = \bs{I} \setminus \langle \ell \rangle$.  Using Bayes' theorem, the conditional joint density for $\bm{W_S}$ is given by
\bea
f_{\bm{W_S},W_\ell|K}(\bm{w_S},w_\ell | \ell) &=&\frac{\PP(K=\ell| \bm{W_S} = \bm{w_S},W_\ell = w_\ell)}{\PP(K=\ell)} f_{\bm{W_S}}(\bm{w_S})f_{W_\ell}(w_\ell) \n \\
&=& \frac{w_\ell/n}{1/(2n)}f_{\bm{W_S}}(\bm{w_S}) \n \\
&=& 2 w_\ell f_{\bm{W_S}}(\bm{w_S}), \quad 0 \le w_\ell \le 1,
\eea
where we have used the results of Lemma \ref{lemmakunif}. This holds for the $K=\ell$ and $\ell = 1,\ldots,n$, so we replace the index $\ell$ with $K$ in the expression.
\end{proof}

\nin We can now analyze the gap obtained by \textsc{Blind-Greedy} for $K>0$. This gives the following lemma and a proof of Theorem \ref{greedythm}.
\begin{lemma}
Independent of the critical item index $K>0$, the conditional distribution of the gap obtained by \textsc{Blind-Greedy} satisfies
\be
f_{G|W_K}(g|w_K) = \mc{U}[0,w_K].
\ee
\label{greedygapcond}
\end{lemma}
\begin{proof}
For any $\ell = 1,\ldots,n$ and any $\bm{W_I} = \bm{w_I}$, the posterior distribution of $B$ given the event $K=\ell$ satisfies
\be
f_{B|\bm{W_I},K}(b|\bm{w_I},\ell) = \mc{U}\left [\sum_{i=1}^{\ell-1}w_i,\sum_{i=1}^{\ell} w_i \right ],
\ee 
since we have a uniform random variable $B$ that is conditionally contained in a given interval. Now using the definition of $G$ in \eqref{gdef},
\be
f_{G|W_\ell,K}(g|w_\ell,\ell) = \mathcal{U}[0,w_\ell].
\ee
\end{proof}
\nin\textit{Proof of Theorem \ref{greedythm}}. Using Lemma \ref{greedygapcond} and the distribution for $W_K$ from Lemma \ref{othersindep}, we have for $K>0$,
\be
f_{G}(g) = \int_0^1 f_{G|W_K}(g|w_K)f_{W_K}(w_K) \dd w_K= \int_{g}^1 \frac{1}{w_K} 2 w_K \mathrm{d} w_K = 2-2g,
\ee
where we have used that $G \le W_K$ with probability one. This serves as a simpler proof of the theorem from \cite{bt91}; their proof is likely more conducive to their analysis.
\qed

Finally, we need a modified version of Lemma \ref{othersindep}, which will be used in the subsequent sections.

\begin{lemma}
Given any critical item $K>0$, gap $G=g$, and any subsequence of items $\bs{S} \subseteq \bs{I} \setminus \lan K \ran$, the weights $\bs{W_S}$ are independently distributed on $\mathcal{U}[0,1]$, and $W_K$ is independently distributed on $\mc{U}[g,1]$.
\label{exhindep}
\end{lemma}
\begin{proof}
Fix $K=\ell$ for any $\ell>0$. The statement of the Lemma is equivalent to the expression
\bea
f_{\bs{W_S},W_\ell|G,K}(\bs{w_S}, w_\ell|g,\ell) = \frac{1}{1-g} f_{\bs{W_S}}(\bs{w_S}),\quad g \le w_\ell \le 1.
\eea
Note that
\bea
f_{G|\bs{W_S}, W_\ell,K}(g|\bs{w_s}, w_\ell, \ell) = \mc{U}[0,w_\ell], \label{jordan}
\eea
which can be shown by the same argument for Lemma \ref{greedygapcond}. Then,
\bea
f_{\bs{W_S},W_\ell|G,K}(\bs{w_S},w_\ell|g,\ell) &=& \frac{f_{G|\bs{W_S}, W_\ell, K}(g|\bs{w_S},w_\ell,\ell) f_{\bs{W_S},W_\ell|K}(\bs{w_S},w_\ell|\ell)}{f_{G|K}(g|\ell)} \n \\
&=& \frac{f_{G|\bs{W_S}, W_\ell, K}(g|\bs{w_S},w_\ell,\ell) f_{\bs{W_S}}(\bs{w_S}) f_{W_\ell|K}(w_\ell|\ell)}{f_{G|K}(g|\ell)} \n \\
&=& \frac{1}{w_\ell} \frac{2w_\ell}{2-2g}  f_{\bs{W_S}}(\bs{w_S}), \quad g \le w_\ell \le 1, 
\eea
where we have used Lemma \ref{othersindep}, \eqref{jordan}, and Theorem \ref{greedythm}.
\end{proof}

\section{Consecutive rollout}
\label{consecutive}

The \textsc{Consecutive-Rollout} algorithm is shown in Algorithm \ref{consecRollout}. The algorithm takes as input a sequence of item weights $\bs{w_I}$ and capacity $b$, and makes calls to \textsc{Blind-Greedy} as a subroutine. At iteration $i$, the algorithm calculates the value ($U_+$) of adding item $i$ to the solution and using \textsc{Blind-Greedy} on the remaining items, and the value  ($U_-$) of not adding the item to the solution and using \textsc{Blind-Greedy} thereafter. The item is then added to the solution only if the former valuation ($U_+$) is larger.

\algsetup{indent=2em}
\begin{algorithm}
\caption{\textsc{Consecutive-Rollout}}
\label{consecRollout}
\begin{algorithmic}[1]
\REQUIRE Item weight sequence $\bs{w_I}$ where $\bs{I} = \langle 1, \ldots,n \rangle$, capacity $b$.
\ENSURE Feasible solution sequence $\bs{S}$, value $U$.
\STATE Initialize $\bs{S} \leftarrow \langle \rangle$, remaining item sequence $\bs{\bar{I}} \leftarrow \bs{I}$, $\bar{b} \leftarrow b$, $U \leftarrow 0$.
\FOR{$i=1$ to $n$ (each item)}
\STATE Estimate the value of adding item $i$, $(\cdot,~U_+)$ = \textsc{Blind-Greedy}($\bs{\bar{I}}$,$~\bar{b}$).
\STATE Estimate the value of skipping item $i$, $(\cdot,~U_-)$ = \textsc{Blind-Greedy}($\bs{\bar{I}} \setminus \langle i\rangle$,$~\bar{b}$).
\IF{$U_{+} > U_{-}$ (estimated value of adding the item is larger)}
\STATE Add item $i$ to solution sequence, $\bs{S} \leftarrow \bs{S}: \lan i \ran$.
\STATE Update remaining capacity, $\bar{b} \leftarrow \bar{b} - w_i$, and value, $U \leftarrow U+p_i$.
\ENDIF
\STATE Remove item $i$ from the remaning item sequence, $\bs{\bar{I}} \leftarrow \bs{\bar{I}} \setminus \lan i \ran$.
\ENDFOR
\STATE Return $\bs{S}$, $U$.
\end{algorithmic}
\end{algorithm}

We only focus on the result of the first iteration of the algorithm; bounds from the first iteration are valid for future iterations\footnote{The technical condition for this property to hold is that the base policy/algorithm is sequentially consistent, as defined in \cite{bertsekas97}. It is easy to verify that \textsc{Blind-Greedy} satisfies this property.}. A single iteration of \textsc{Consecutive-Rollout} effectively takes the best of two solutions, the one obtained by \textsc{Blind-Greedy} and the solution obtained from using \textsc{Blind-Greedy} after removing the first item. Let $V_*(n)$ denote the gap obtained by a single iteration of the rollout algorithm for the subset sum problem with $n$ items under the stochastic model.
\begin{theorem}
For the subset sum problem with $n\ge3$, the gap $V_*(n)$ obtained by running a single iteration of \textsc{Consecutive-Rollout} satisfies
\be
\E \left [ V_*(n) \left \vert \sum_{i=1}^nW_i>B \right. \right ] \le \frac{3+13n}{60n} \le \frac{7}{30} \approx 0.233. \label{ss_consec_eq}
\ee
\label{ss_consec_thm}
\end{theorem}

\begin{figure}
\begin{center}
\includegraphics[width=0.48\textwidth]{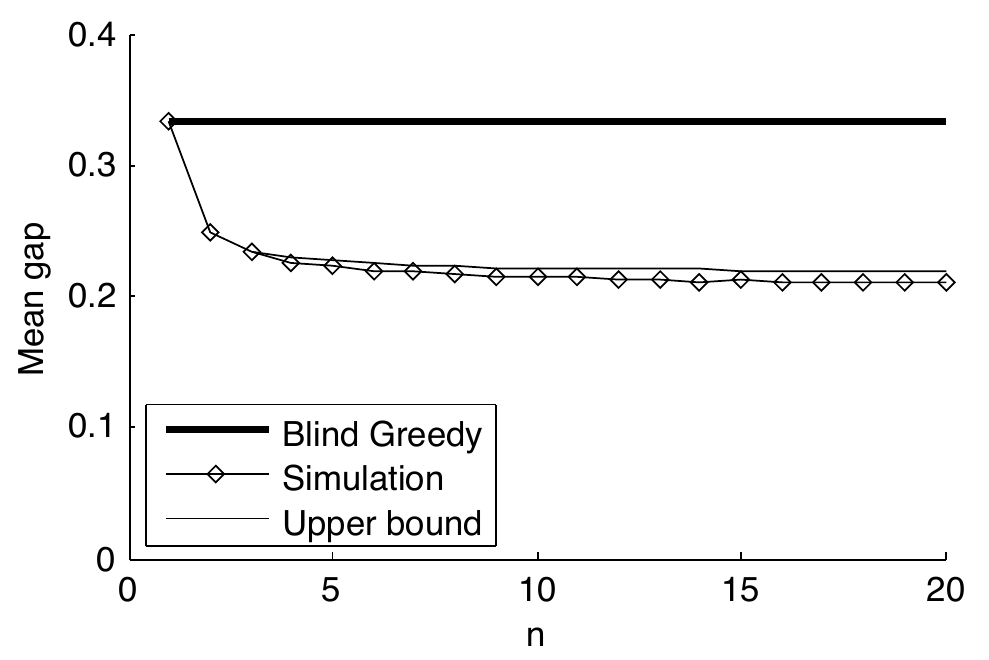}
\includegraphics[width=0.48\textwidth]{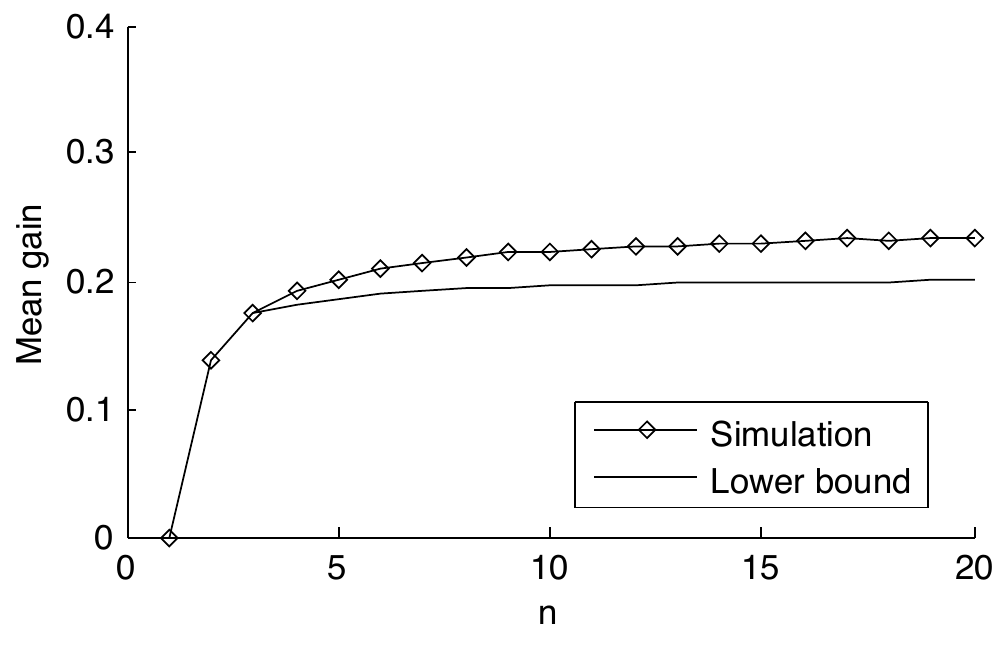}
\caption{Performance bounds and simulated values for the expected gap $\E[V_*(n)|\cdot]$ and expected gain $\E[Z_*(n)|\cdot]$ after running a single iteration of the \textsc{Consecutive-Rollout} algorithm on the subset sum problem and knapsack problem, respectively. For each $n$, the mean values are shown for $10^5$ simulations.\label{consecplots}}
\end{center}
\end{figure}
\nin As expected, there is not a strong dependence on $n$ for this algorithm. The bound is tight for $n=3$, where it evaluates to $\frac{7}{30} \approx 0.233$. It is also clear that $\lim_{n\rightarrow \infty} \E[V_*(n)|\cdot] \le \frac{13}{60} \approx 0.217$. The bounds are shown with simulated performance in Figure \ref{consecplots}(a). A similar result holds for the knapsack problem.
\begin{theorem}For the knapsack problem with $n\ge3$, the gain $Z_*(n)$ obtained by running a single iteration of \textsc{Consecutive-Rollout} satisfies
\be
\E \left [ Z_*(n) \left \vert \sum_{i=1}^nW_i>B \right. \right ] \ge \frac{-26+59n}{288n} \ge \frac{151}{864} \approx 0.175.
\ee
\label{kp_consec_thm}
\end{theorem}
\nin The bound is plotted with simulated values in Figure \ref{consecplots}(b). Again the bound is tight for $n=3$ with a gain of $\frac{151}{864}\approx 0.175$. Asymptotically, $\lim_{n\rightarrow \infty} \E[Z_*(n)|\cdot] \ge \frac{59}{288}\approx 0.205$. The rest of this section is devoted to the proof of Theorem \ref{ss_consec_thm}. The proof of Theorem \ref{kp_consec_thm} follows a similar structure and is given in the supplementary material.

\subsection{Consecutive rollout: subset sum problem analysis}

The proof idea for Theorem \ref{ss_consec_thm} is to visually analyze the solution sequence given by \textsc{Blind-Greedy} on the nonnegative real line, as shown in Figure \ref{vis}, and then look at modifications to this solution caused by removing the first item. Removing the first item causes the other items to slide to the left and may make some remaining items feasible to pack. We determine bounds on the gap produced by this procedure while conditioning on the greedy gap $G$, critical item $K$, and the item weights $(W_K,W_{K+1})$. We then take the minimum of this gap and the greedy gap and integrate over conditioned variables to obtain the final bound. Our analysis is divided into lemmas based on the critical item $K$. We show a detailed proof of the lemma corresponding to $2 \le K \le n-1$. For the cases where $K=1$ or $K=n$, the proofs are similar and are placed in the supplementary material.

To formalize the behavior of \textsc{Consecutive-Rollout}, we introduce the following two definitions. The \textit{drop critical item} $L_1$ is the index of the item that becomes critical when the first item is removed and thus satisfies
\bea
\left \{
\begin{array}{ll}
\ds \sum_{i = 2}^{L_1-1} W_i \le B < \sum_{i = 2}^{L_1} W_i & \qquad \sum_{i=2}^n W_i > B \n \\
L_1 = n+1 & \qquad \sum_{i=2}^n W_i \le B,
\end{array}
\right.
\eea
where the latter case signifies that all remaining items can be packed. The \textit{drop gap} $V_1$ then has definition
\bea
V_1 := B - \sum_{i = 2}^{L_1-1} W_i.
\eea
We are ultimately interested in the minimum of the drop gap and the greedy gap, which we refer to as the \textit{minimum gap}, and is the value obtained by the first iteration of the rollout algorithm:
\bea
V_*(n) := \min(G,V_1).
\eea
We will often write write $V_*(n)$ simply as $V_*$. We will also use $\mc{C}_i$ to denote the event that item $i$ is critical and $\overline{\mc{C}_{1n}}$ for the event that $2 \le K \le n-1$. Also recall that we have $\bs{P_I} = \bs{W_I}$ for the subset sum problem.

\begin{lemma}
\label{sscbc1n}
For  $2\le K \le n-1$, the expected minimum gap satisfies
\be
\E[V_*(n)|2 \le K \le n-1] \le \frac{13}{60}.
\ee
\end{lemma}
\begin{proof}
Fix $K=k$ for $2\le k \le n-1$. The drop gap in general may be a function of the weights of all remaining items. To make things more tractable, we define the random variable $V_1^u$ that satisfies $V_1 \le V_1^u$ with probability one, and as we will show, is a deterministic function of only $(G, W_1, W_k, W_{k+1})$. The variable $V_1^u$ is specifically defined as

\bea
V_1^u := \left \{
\begin{array}{ll}
V_1 & L_1 = k \vee L_1 = k+1 \\
B - \sum_{i = 2}^{k+1}W_i & L_1 \ge k+2.
\end{array}
\right.
\eea
In effect, $V_1^u$ does not account for the additional reduction in the gap given if any of the items $i \ge k+2$ become feasible, so it is clear that $V_1^u \le V_1$.

To determine the distribution of $V_1^u$, we start by considering scenarios where $L_1 \ge k+2$ is not possible and thus $V_1^u = V_1$. For $G = g$ and $\bs{W_I} = \bs{w_I}$, an illustration of the drop gap as determined by $(g,w_1,w_k,w_{k+1})$ is shown in Figure \ref{C2}. We will follow the convention of using lowercase letters for random variables shown in figures and when referring these variables. The knapsack is shown at the top of the figure with items packed from left to right, and at the bottom the drop gap $v_1$ is shown as a function of $w_1$. The shape of the function is justified by considering different sizes of $w_1$. As long as $w_1$ is smaller than $w_k-g$, the gap given by removing the first item increases at unit rate. As soon as $w_1= w_k-g$, item $k$ becomes feasible and the gap jumps to zero. The gap then increases at unit rate and another jump occurs when $w_1$ reaches $w_k-g+w_{k-1}$. The case shown in the figure satisfies $w_k - g + w_{k+1}+w_{k+2}>1$. It can be seen that this is a sufficient condition for the event $L_1 \ge k+2$ to be impossible since even if $w_1 = 1$, item $k+2$ cannot become feasible. It is for this reason that $v_1$ is uniquely determined by $(g, w_1, w_k,w_{k+1})$ here.

\begin{figure}[h]
\begin{center}
\includegraphics[scale=0.7] {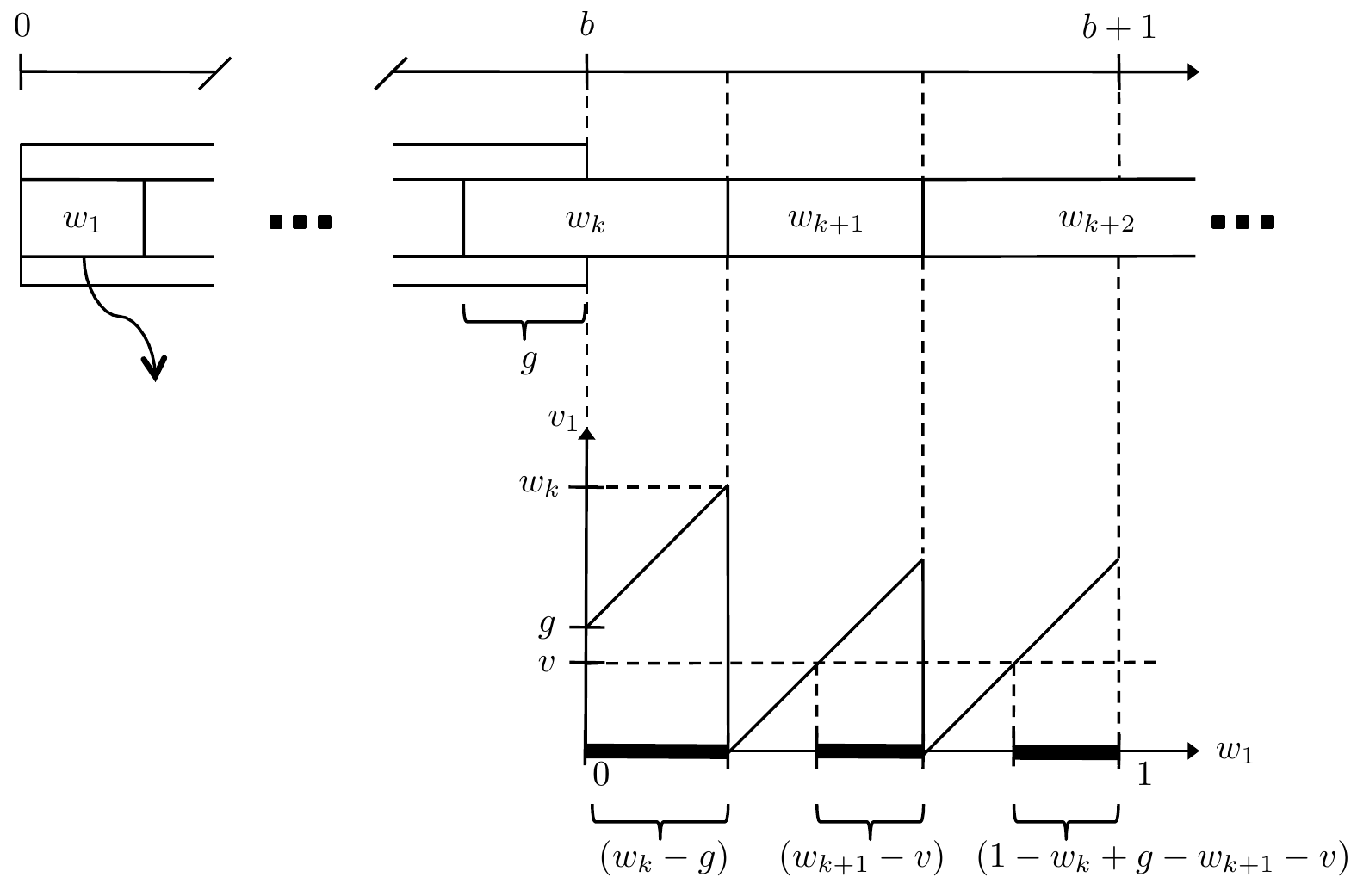}
\caption{Gap $v_1$ as a function of $w_1$, parameterized by $(g,w_k, w_{k+1})$, resulting from the removal of the first item and assuming that $K=k$ with $2 \le k \le n-1$. The function starts at $g$ and increases at unit rate, except at $w_1=w_k-g$ and $w_1=w_k-g+w_{k+1}$, where the function drops to zero. If we only condition on $(g,w_k,w_{k+1})$, the probability of the event $V_1>v$ is given by the total length of the bold regions for $v < g$. Note that in the figure, $w_k-g+w_{k+1} < 1$ and the second two bold segments have positive length; these properties do not hold in general.  \label{C2}}
\end{center}
\end{figure}

Continuing with the case shown in the figure, if we only condition on $(g,w_k, w_{k+1})$, we have by Lemma \ref{exhindep} that $W_1$ follows distribution $\mc{U}[0,1]$, meaning that the event $V_1 > v$ is given by the length of the bold regions on the $w_1$ axis. We explicitly describe the size of these regions. Assuming that $L_1 \le k+1$, we derive the following expression:
\bea
\PP(V_1>v |g,w_k, w_{k+1},\overline{\mc{C}_{1n}}, L_1 \le k+1) &=& (w_k-g) + (w_{k+1}-v)_+ +(1-w_k+g-w_{k+1}-v)_+ \n \\
&& - (w_k-g+w_{k+1}-1)_+, \quad  v < g.
\label{babymarty}
\eea
The first three terms in the expression come from the three bold regions shown in Figure \ref{C2}. We have specified that $v<g$, so the length of the first segment is always $w_k-g$. For the second term, it is possible that $v>w_{k+1}$, so we only take the positive portion of $w_{k+1}-v$. In the third term, we take the positive portion to account for the cases where (1) item $k+1$ does not become feasible, meaning $w_k-g+w_{k+1}>1$, and (2) if it is feasible, where $v$ is greater than the height of the third peak, where $v > 1-w_k+g-w_{k+1}$. 
 
The last term is required for the case where item $k+1$ does not become feasible, as we must subtract the length of the bold region that potentially extends beyond $w_1 = 1$. Note that we always subtract one in this expression since it is not possible for the $w_1$ value where $v_1  = v$ on the second peak to be greater than one. To see this, assume the contrary, so that $v+w_k-g>1$. This inequality is obtained since on the second peak we have $v_1 = g-w_k+w_1$ and the $w_1$ value that satisfies $v_1 = v$ is equal to $v+w_k-g$. The statement $v+w_k-g > 1$, however, violates our previously stated assumption that $g < v$.

We now argue that we in fact have $V_1 \le V_1^u$ with probability one, where
\bea
\PP(V_1^u >v |g,w_k, w_{k+1},\overline{\mc{C}_{1n}}) &=& (w_k-g) + (w_{k+1}-v)_+ +(1-w_k+g-w_{k+1}-v)_+ \n \\
&& - (w_k-g+w_{k+1}-1)_+, \quad v < g.
\label{marty}
\eea
We have simply replaced $V_1$ with $V_1^u$ in \eqref{babymarty} and removed the condition $ L_1 \le k+1$. We already know that the expression is true for $L_1 \le k+1$. For $L_1 \ge k+2$, we refer to Figure \ref{C2} and visualize the effect of a much smaller $w_{k+2}$, so that $w_k - g + w_{k+1}+w_{k+2}<1$. This would yield four (or more) peaks in the $v_1$ function. To determine the probability of the event $V_1>v$ while $W_1$ is random, we would have to evaluate the sizes of these extra peaks, which would be a function of $w_{k+2}$, $w_{k+3}$, etc. However, our definition of $V_1^u$ does not account for the additional reductions in the gap given by items beyond $k+1$. We have already shown that $V_1 \le V_1^u$, and it is now clear that $V_1^u$ is a deterministic function of $(G, W_1, W_k, W_{k+1})$, and that \eqref{marty} is justified.


We now evaluate the minimum of $V_1^u$ and $G$ and integrate over the conditioned variables. To begin, note that conditioning on the gap $G$ makes $V_1^u$ and $G$ independent, so,
\bea
\PP(V_1^u>v, G > v |\overline{\mc{C}_{1n}},g,w_k, w_{k+1}) &=& \PP(V_1^u>v |\overline{\mc{C}_{1n}},g,w_k, w_{k+1})\I(v<g).
\eea
Marginalizing over $W_{k+1}$, which has uniform density according to Lemma \ref{exhindep}, gives
\bea
\PP(V_1^u > v, G>v|\overline{\mc{C}_{1n}},g,w_k) &=& \int_0^1\PP(V_1^u > v, g>v|\overline{\mc{C}_{1n}},g,w_k,w_{k+1})f_{W_{k+1}}(w_{k+1})\mathrm{d} w_{k+1} \n \\
&=&\left ( (w_k-g) + \frac{1}{2}(1-v)^2 - \frac{1}{2}(w_k-g)^2  \right. \n \\ 
&&\left . +\frac{1}{2}(1-w_k+g-v)_+^2 \right ) \I(v<g).
\eea
Using Lemma \ref{greedygapcond}, we have
\bea
\PP(V_1^u>v, G>v|\overline{\mc{C}_{1n}},w_k) &=& \int_0^{w_k} \PP(V_1^u > v, G>v|\overline{\mc{C}_{1n}},g,w_k) f_{G|\overline{\mc{C}_{1n}},W_k}(g|\overline{\mc{C}_{1n}},w_k) \dd g \n \\
&=& 1-2v-\frac{v}{w_k}+\frac{2v^2}{w_k}-\frac{v^3}{2 w_k}+\frac{v w_k}{2}.
\eea
Finally, we integrate over $W_k$ according to Lemma \ref{othersindep}
\bea
\PP(V_1^u > v, G>v|\overline{\mc{C}_{1n}}) &\le&  \int_v^1\PP(V_1^u>v,G>v|\overline{\mc{C}_{1n}},w_k)f_{W_k}(w_k)\dd w_k \n \\
&=& 1-\frac{11v}{3}+5v^2-3v^3+\frac{2v^4}{3}.
\eea
This term is sufficient for calculating the expected value bound.
\end{proof}

\begin{lemma}
\label{ssccn}
For $K = n$, the expected minimum gap satisfies
\be
\E[V^*(n)|K=n] = \frac{1}{4}.
\ee
\end{lemma}
\nin \textit{Proof.} Supplementary material.

\begin{lemma}
\label{sscc1}
For $K=1$, the expected minimum gap satisfies
\be
\E[V^*(n)|K=1] \le \frac{7}{30}.
\ee
\end{lemma}
\nin \textit{Proof.} Supplementary material.

The final result for the subset sum problem follows easily from the stated lemmas. ~\\

\nin\textit{Proof of Theorem \ref{ss_consec_thm}}
~Using the above Lemmas and noting that the events $\mc{C}_1$, $\overline{\mc{C}_{1n}}$, and $C_n$ form a partition of the event $\sum_{i\in I} W_i > B$, the result follows using the total expectation theorem and Lemma \ref{lemmakunif}. \qed


\section{Exhaustive rollout}
\label{exhaustive}
The \textsc{Exhaustive-Rollout} algorithm is shown in Algorithm \ref{exhRollout}. It takes as input a sequence of item weights $\bs{w_I}$ and capacity $b$. At each iteration, indexed by $t$, the algorithm considers all items in the available sequence $\bs{\bar{I}}$. It calculates the value obtained by moving each item to the front of the sequence and applying the \textsc{Blind-Greedy} algorithm. The algorithm then adds the item with the highest estimated value (if it exists) to the solution. We implicitly assume a consistent tie-breaking method, such as giving preference to the item with the lowest index. The next iteration then proceeds with the remaining sequence of items.

\algsetup{indent=2em}
\begin{algorithm}
\caption{\textsc{Exhaustive-Rollout}}
\label{exhRollout}
\begin{algorithmic}[1]
\REQUIRE Item weight sequence $\bs{w_I}$ where $\bs{I} = \langle 1, \ldots,n \rangle$, capacity $b$.
\ENSURE Feasible solution sequence $\bs{S}$, value $U$.
\STATE Initialize $\bs{S} \leftarrow \langle \rangle$, $\bs{\bar{I}} \leftarrow \bs{I}$, $\bar{b} \leftarrow b$, $U \leftarrow 0$.
\FOR{$t=1$ to $n$}
\FOR{$i\in \bs{\bar{I}}$ (each item in remaning item sequence)}
\STATE Let $\bs{\bar{I}^i}$ denote the sequence $\bs{\bar{I}}$ with $i$ moved to the first position.
\STATE Estimate value of sequence, $(\cdot,U_i)$ = \textsc{Blind-Greedy}($\bs{w_{\bar{I}^i}}$,$~\bar{b}$).
\ENDFOR
\IF{$\max_i U_i > 0$}
\STATE Determine item with max estimated value, $i^* \leftarrow \argmax_i U_i$.
\STATE Add item $i^*$ to solution sequence, $\bs{S} \leftarrow \bs{S} : \lan i^* \ran$, $\bs{\bar{I}} \leftarrow \bs{\bar{I}} \setminus \lan i^* \ran$.
\STATE Update remaining capacity, $\bar{b} \leftarrow \bar{b} - w_i$, and value, $U \leftarrow U+p_i$.
\ENDIF
\ENDFOR
\STATE Return $S$, $U$.
\end{algorithmic}
\end{algorithm}


\begin{figure}
\begin{center}
\begin{tabular}{cc}
\includegraphics[width=0.47\textwidth]{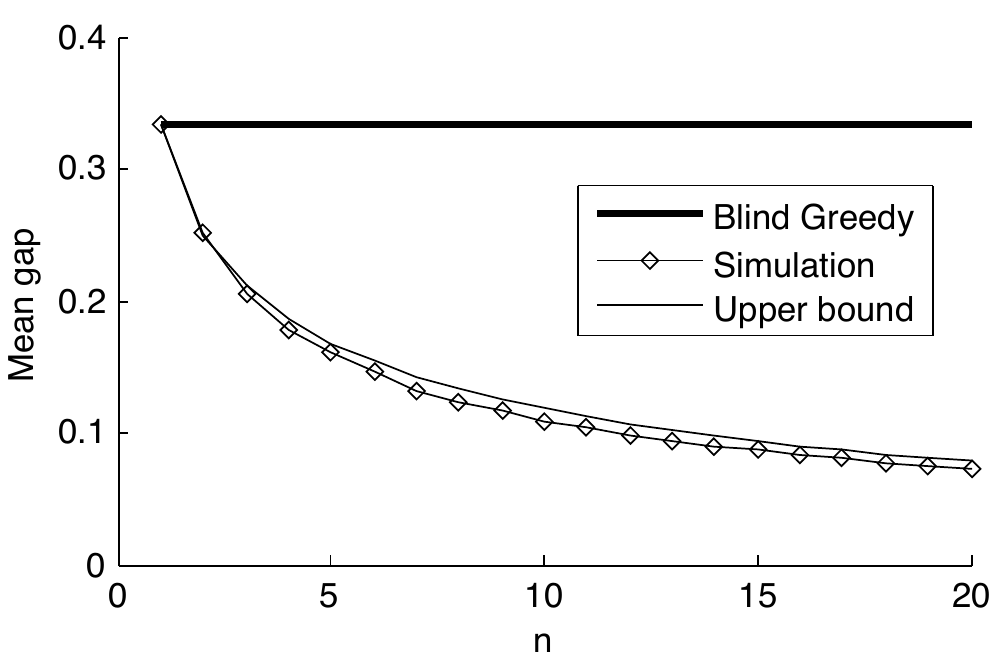} &
\includegraphics[width=0.47\textwidth]{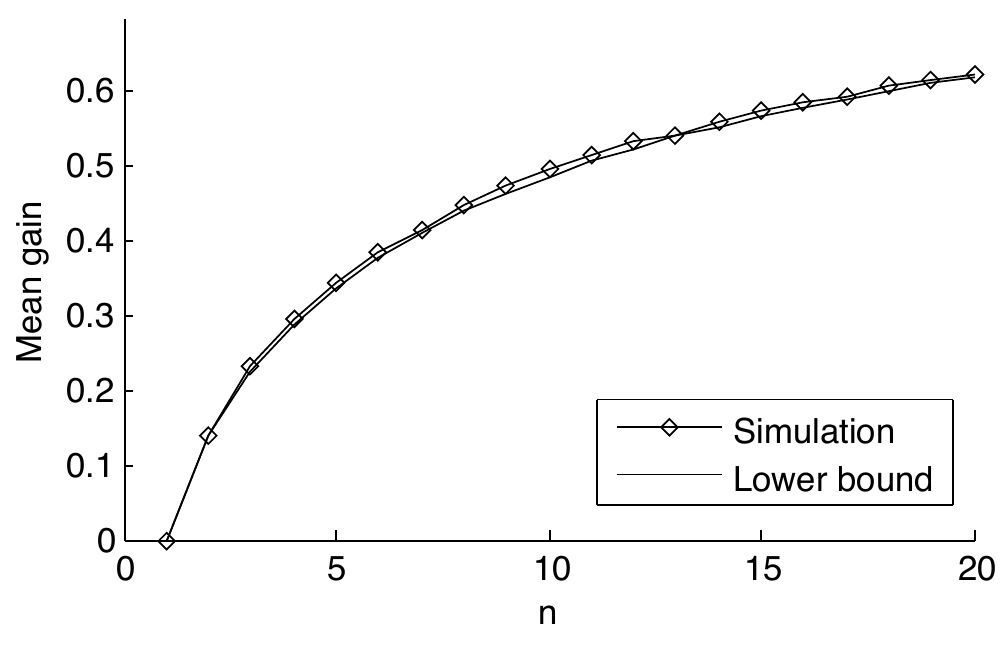} \\
\small (a) & \small (b)
\end{tabular}
\caption{Performance bounds and simulated values for (a) expected gap $\E[V_*(n)|\cdot]$ and (b) expected gain $\E[Z_*(n)|\cdot]$ after running a single iteration of \textsc{Exhaustive-Rollout} on the subset sum problem and knapsack problem, respectively. For each $n$, the mean values are shown for $10^5$ simulations. \label{exhplots}}
\end{center}
\end{figure}

We again only consider the first iteration, which tries using \textsc{Blind-Greedy} after moving each item to the front of the sequence, and takes the best of these solutions. This gives an upper bound for the subset sum gap and a lower bound on the knapsack problem gain following from additional iterations. For the subset sum problem, let $V_*(n)$ denote the gap obtained after a single iteration of \textsc{Exhaustive-Rollout} on the stochastic model with $n$ items. We have the following bounds.
\begin{theorem}
For the subset sum problem, the gap $V_*(n)$ after running a single iteration of \textsc{Exhaustive-Rollout} satisfies
\bea
\E \left [ V_*(n) \left \vert \sum_{i=1}^nW_i>B \right. \right ] &\le& \frac{1}{n(2+n)} + \frac{1}{n} \sum_{m=0}^{n-2}  \frac{9+2 m}{3(3+m)(4+m)}.
\eea \label{exhthm}
\end{theorem}
\begin{corollary}
\bea
\E \left [ V_*(n) \left \vert \sum_{i=1}^nW_i>B \right. \right ] &\le& \frac{1}{n (2+n)}+\frac{1}{n}\log \left[ \left (\frac{3+2n}{5} \right )  \left (\frac{7}{5+2n} \right )^{1/3} \right] \label{exup}.
\eea \label{exhcoro}
\end{corollary}
\begin{theorem} \bea
\lim_{n \rightarrow \infty} \E \left [ V_*(n) \left \vert \sum_{i=1}^nW_i>B \right. \right ] = 0, \qquad \E \left [ V_*(n) \left \vert \sum_{i=1}^nW_i>B \right. \right ] = O \left ( \frac{\log n}{n} \right ).
\eea \label{ssasym}
\end{theorem}
\nin A plot of the bounds and simulated results is shown in Figure \ref{exhplots}(a). For the knapsack problem, let $Z_*(n)$ denote the gain given by a single iteration of \textsc{Exhaustive-Rollout}. The expected gain is bounded by the two following theorems, where $H(n)$ is the $n$th harmonic number, $H(n) = \sum_{\ell=1}^n \frac{1}{\ell}$.
\begin{theorem}
For the knapsack problem, the gain $Z_*(n)$ after running a single iteration \textsc{Exhaustive-Rollout} satisfies
\bea
&& \E \left [Z_*(n) \left \vert \sum_{i=1}^nW_i > B \right. \right ] \ge 1 + \frac{2}{n(n+1)} - \frac{2H(n)}{n^2} \n \\
&&+\frac{1}{n}\sum_{m=0}^{n-2} \left ( \sum_{j=1}^{m+1} T(j,m) + \left( (186+472m + 448m^2 + 203m^3 + 45 m^4 + 4m^5) \right. \right. \n \\
&& -(244 + 454m + 334m^2 + 124m^3 + 24m^4 + 2m^5) H(m+1) \n \\
&& \left. \left . -(48 + 88m + 60m^2 + 18m^3 + 2m^4) (H(m+1))^2 \right ) \frac{1}{(m+1)(m+2)^3(m+3)^2}  \right ), \n \\
\eea
where
\bea
T(j,m) &:=& \frac{2\left(-4+j-4 m+j m-m^2-\left(j+(2+m)^2\right) H(j)\right)}{j (-3+j-m) (-2+j-m) (1+m) (2+m)}\n \\
&&+\frac{2(j+(2+m)^2)H(3+m)}{j (-3+j-m) (-2+j-m) (1+m) (2+m)}.
\eea
\label{kpexhthm}
\end{theorem}
\begin{theorem}
\bea
\lim_{n \rightarrow \infty} \E \left [Z_*(n) \left \vert \sum_{i=1}^nW_i>B \right. \right ] = 1, \qquad 1 - \E \left [ Z_*(n) \left \vert \sum_{i=1}^nW_i>B \right. \right ]  = O\left (\frac{\log^2 n}{n} \right ).
\eea
\label{kpasym}
\end{theorem}
\nin The expected gain approaches unit value at a rate slightly slower than the convergence rate for the subset sum problem\footnote{This is likely a result of the fact that in the subset sum problem, the algorithm is searching  for an item with one criteria: a weight approximately equal to the gap. For the knapsack problem, however, the algorithm must find an item satisfying two criteria: a weight smaller than the gap and a profit approximately equal to one.}. The gain is plotted with simulated values in Figure \ref{exhplots}(b). While the bound in Theorem \ref{kpexhthm} does not admit a simple integral bound, omitting the nested summation term  $\sum_{j=1}^m T(j,m)$ gives a looser but valid bound. We show the proof of Theorem \ref{exhthm} in the remainder of this section. All remaining results are given in the supplementary material.

\subsection{Exhaustive rollout: subset sum problem analysis}
The proof method for Theorem \ref{exhthm} is similar to the approach taken in the previous section. With Figure \ref{vis} in mind, we will analyze the effect of individually moving each item to the front of the sequence, which will cause the other items to shift to the right. Our strategy is to perform this analysis while conditioning on three parameters: the greedy gap $G$, the critical item $K$, and the weight of the last packed item $W_{K-1}$. We then find the minimum gap given by trying all items and integrate over conditioned variables to obtain the final bound.


To analyze solutions obtained by using \textsc{Blind-Greedy} after moving a given item to the front of the sequence, we introduce two definitions. The \textit{$j$th insertion critical item} $L_j$ is the first item that is infeasible to pack by \textsc{Blind-Greedy} when item $j$ is moved to the front of the sequence. Equivalently, $L_j$ satisfies
\bea
\left \{
\begin{array}{ll}
\ds W_j + \sum_{i=1}^{L_j-1} W_i~\I(i \ne j) \le B < W_j + \sum_{i=1}^{L_j} W_i~\I(i \ne j) & \qquad W_j \le B \\
L_j = j & \qquad W_j > B.
\end{array}
\right.
\eea
We then define the corresponding \textit{$j$th insertion gap} $V_j$, which is the gap given by the greedy algorithm when item $j$ is moved to the front of the sequence:
\bea
V_j := B - \I(W_j \le B) \left (  W_j + \sum_{i = 1}^{L_j-1} W_i ~ \I(i \ne j) \right).
\eea

In the following three lemmas, we bound the probability distribution of the insertion gap for packed items ($j\le K-1$), he critical item ($j = K$), and the remaining items ($j \ge K+1$), while assuming that $K>1$. Lemma \ref{c1} then handles the case where $K=1$. Thereafter we bound the minimum of these gaps and the greedy gap $G$, and finally integrate over the conditioned variables to obtain the bound on the expected minimum gap. The key analysis is illustrated in the proof of Lemma \ref{remain insert}; the related proofs of Lemma \ref{crit insert} and Lemma \ref{c1} are given in the supplementary material. The event $\mc{C}_j$ again indicates that item $j$ is critical, and $\bar{\mc{C}_1}$ indicates the event that the first item is not critical.

\begin{lemma}
For $K>1$ and $j = 1,\ldots,K-1$, the $j$th insertion gap satisfies
\bea
V_j = G
\eea
with probability one.
\label{pack insert}
\end{lemma}
\begin{proof}
This follows trivially since the term $\sum_{i=1}^{K-1}W_i$ in \eqref{critdef} does not depend on the order of summation.
\end{proof}


\begin{lemma}
\label{remain insert}
For $K>1$ and $j = K+1,\ldots,n$, the $j$th insertion gap satisfies $V_j \le V_j^u$ with probability one, where $V_j^u$ is a deterministic function of $(G, W_{K-1}, W_j)$ and conditioning only on $(G ,W_{K-1})$ gives
\bea
\PP(V^u_j > v|g,w_{K-1}, \overline{\mc{C}_1}) &=& (g-v)_+ + (w_{K-1}-v)_+ - (g+w_{K-1}-v-1)_+ + (1-g-w_{K-1})_+ \label{torri} \n \\
&=:& \PP(V^u > v|g,w_{K-1}, \overline{\mc{C}_1}).
\eea
\end{lemma}

\begin{proof}
Fix $K=k$ for $k>1$. To simplify notation make the event $\overline{\mc{C}_1}$ implicit throughout the proof.  Define the random variable $V_j^u$ so that 
\bea
V_j^u = \left \{
\begin{array}{ll}
V_j & L_j = k \vee L_j = k-1 \n \\
1 & L_j \le k-2 \vee L_j = j.
\end{array}
\right.
\eea
While $V_j$ may in general depend on $(G, W_j, W_1,\ldots,W_{k-1})$, the variable $V_j^u$ is chosen so that it only depends on $(G, W_{k-1}, W_j)$. In cases where $V_j$ does only depend on $(G,W_{k-1},W_j)$, we have $V_j^u = V_j$. When $V_j$ depends on more than these three variables, $V_j^u$ assumes a worst-case bound of unit value.

We begin by analyzing the case where $L_j = k \vee L_j = k-1$ so that the insertion gap $V_j$ is equal to $V_j^u$. For $G=g$ and $\bs{W_I} = \bs{w_I}$, a diagram illustrating the insertion gap as determined by $g$, $w_{k-1}$, and $w_j$ is shown in Figure \ref{eg}. The knapsack is shown at the top of the figure with items packed sequentially from left to right. The plot at the bottom shows the insertion gap $V_j$ that occurs when item $j$ is inserted at the front of the sequence, causing the remaining packed items to slide to the right. The plot is best understood by visualizing the effect of varying sizes of $w_j$. If $w_j$ is very small, the items slide to the right and reduce the gap by the amount $w_j$. Clearly if $w_j = g$ then $v_j = 0$ as indicated by the function. As soon as $w_j$ is slightly larger than $g$, it is infeasible to pack item $k-1$ and the gap jumps. Thus for the instance shown, the $j$th insertion gap is a deterministic function of $(g, w_{k-1}, w_j)$.

\begin{figure}[h]
\begin{center}
\includegraphics[scale=0.7] {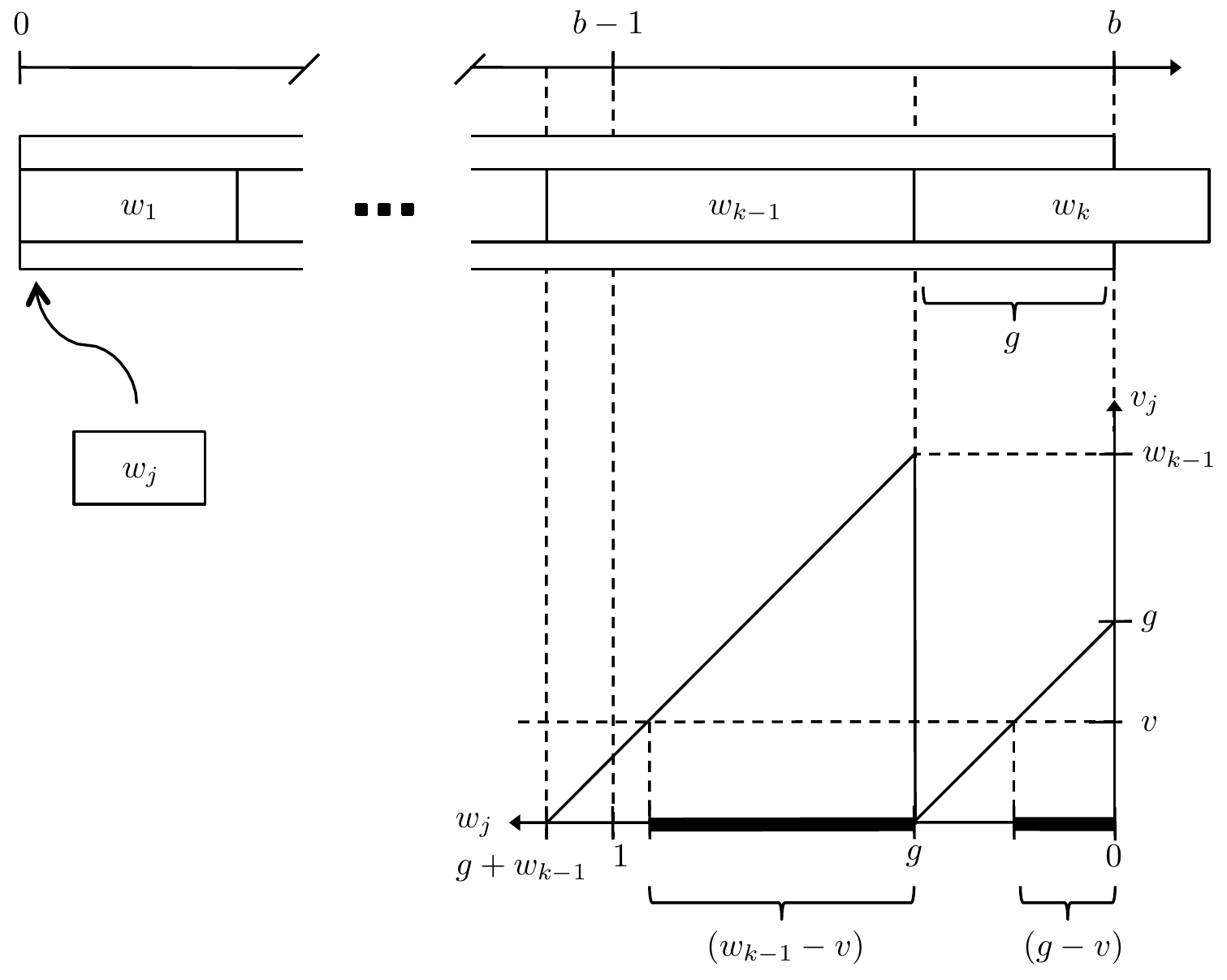}
\caption{Insertion gap $v_j$ as a function of $w_j$, parameterized by $(w_{k-1},g)$. The function starts at $g$ and decreases at unit rate, except at $w=g$ where the function jumps to value $w_{k-1}$. The probability of the event $V_j > v$ conditioned only on $w_{k-1}$ and $g$ is given by the total length of the bold regions, assuming that $v < g$ and $g+w_{k-1}-v \le 1$. Based on the sizes of $g$ and $w_{k-1}$ shown, only the events $L_j = k$ and $L_j = k-1$ are possible. \label{eg}}
\end{center}
\end{figure}

Considering the instance in the figure, if we only condition on $g$ and $w_{k-1}$ and allow $W_j$ to be random, then $V_j$ becomes a random variable whose only source of uncertainty is $W_j$. Since by Lemma \ref{exhindep} $W_j$ has distribution $\mc{U}[0,1]$, the probability of the event $V_j > v$ is given by the length of the bold regions on the $w_j$ axis. 

We now explicitly describe the length of the bold regions for all cases of $w_{k-1}$ and $g$; this will include the case $L_j = k-2~\vee~L_j= j$ (not possible for the instance in the figure), so the length of the bold regions will define $V_j^u$. Starting with the instance shown, we have $\PP(V_j^u>v|g,w_{k-1}) =\PP(V_j>v|g,w_{k-1}) = (g-v) + (w_{k-1}-v)$ as given by the lengths of the two bold regions, corresponding to the events $L_j = k$ and $L_j = k-1$, respectively. This requires that $v\le g$ and $v \le w_{k-1}$, so the expression becomes $\PP(V_j^u>v|g,w_{k-1}) = (g-v)_+ + (w_{k-1}-v)_+$. We must account for the case where $g+w_{k-1}-v>1$, requiring that we subtract length $(g+w_{k-1}-v-1)$, so we revise the expression to $\PP(V_j^u>v|g,w_{k-1}) = (g-v)_+ + (w_{k-1}-v)_+-(g+w_{k-1}-v-1)_+$. Finally, for the case of $g+w_{k-1}<1$, we must take care of the region where $w_i \in [g+w_{k-1},1]$. It is at this point that the event $L_j \le k-2$ or $L_j = j$ becomes possible and the distinction between $V_j^u$ and $V_j$ is made. Here we have by definition $V_j^u = 1$, which trivially satisfies $V_j \le V_j^u$, so for any $0 \le v < 1$ this region contributes $(1-g-w_{k-1})$ to $\PP(V_j^u>v|g,w_{k-1}).$ This is handled by adding the term $(1-g-w_{k-1})_+$ to the expression. We finally arrive at
\bea
\PP(V_j^u> v|g,w_{k-1}) &=& (g-v)_+ + (w_{k-1}-v)_+ - (g+w_{k-1}-v-1)_+ + (1-g-w_{k-1})_+.
\eea
This holds true for any fixed $k$ as long as $k >1$, so we may replace $w_{k-1}$ with $w_{K-1}$ and make the event $\overline{\mc{C}_1}$ explicit to obtain the statement of the lemma.
\end{proof}

\begin{lemma}
\label{crit insert}
For $K>1$, the $K$th insertion gap satisfies $V_K \le V_K^u$ with probability one, where $V_K^u$ is a deterministic function of $(G, W_{K-1}, W_K)$ and conditioning only on $(G, W_{K-1})$ gives
\bea
\PP(V_K^u>v|g,w_{K-1},\overline{\mc{C}_1}) &=& \left (\frac{1}{1-g} \right) \left ( (w_{K-1}-v)_+-(g+w_{K-1}-v-1)_++(1-g-w_{K-1})_+ \right ) \n \\
&=:& \PP(\widetilde{V}^u>v|g,w_{K-1},\overline{\mc{C}_1}).
\eea
\end{lemma}
\nin \textit{Proof.} Supplementary material.


\begin{lemma}
\label{c1}
For $K=1$ and $j = 2,\ldots,n$, the $j$th insertion gap is a deterministic function of $(W_j, G)$, and conditioning only on $G$ gives
\bea
\PP(V_j>v | g, \mc{C}_1) = (1-v)\I(v < g).
\eea
\end{lemma}
\nin \textit{Proof.} Supplementary material.


Recall that $V_*(n)$ is the gap obtained after the first iteration of the rollout algorithm on an instance $n$ items, which we refer to as the \textit{minimum gap},
\bea
V_*(n) := \min(V_1,\ldots,V_n).
\eea
We will make the dependence on $n$ implicit in what follows so that $V_* = V_*(n)$. We may now prove the final result.

\nin\textit{Proof of Theorem \ref{exhthm}}.
For $K=k>1$, we have $V_* \le V_*^u$ with probability one, where
\bea
V_*^u &:=& \min(G, V_k^u,V_{k+1}^u,\ldots,V_n^u).
\eea
This follows from Lemmas \ref{pack insert} - \ref{crit insert}, as $V_j = G$ for $j \le k-1$. From the analysis in Lemmas \ref{remain insert} and \ref{crit insert}, for each $j \ge k$, $V_j^u$ is a deterministic function of $(G,W_{k-1}, W_k,W_j)$. Furthermore from Lemma \ref{exhindep}, the item weights $W_j$ for $j \ge k+1$ are independently distributed on $\mc{U}[0,1]$, and $W_k$ is independently distributed on $\mc{U}[g,1]$. Thus, conditioning only on $G$ and $W_{k-1}$ makes $V_j^u$ independent for  $j \ge k$, and by the definition of the minimum function, 
\bea
\PP(V_*^u > v | g,w_{k-1},k,\overline{\mc{C}_1}) &=& \PP(G>v|g,w_{k-1},\overline{\mc{C}_1}) \PP(V_k^u >v|g,w_{k-1},\overline{\mc{C}_1}) \prod_{j=k+1}^n \PP(V_j^u >v|g,w_{k-1},\overline{\mc{C}_1}) \n \\
&=& \PP(G>v|g,w_{k-1},\overline{\mc{C}_1})\PP(\widetilde{V}^u >v|g,w_{k-1},\overline{\mc{C}_1}) \left ( \PP(V^u >v|g,w_{k-1},\overline{\mc{C}_1})  \right )^{(n-k)}. \n \\
\eea
Marginalizing over $W_{k-1}$ and $G$ using Lemma \ref{exhindep} and Theorem \ref{greedythm},
\bea
\PP(V_*^u>v|k,\overline{\mc{C}_1}) &=& \int_0^1 \int_0^1 \PP(V^u_*>v|g,w_{k-1},k,\overline{\mc{C}_1}) f_{w_{k-1}}(w_{k-1}) f_G(g) \dd w_{k-1} \dd g \label{theintupper}.
\eea
We refer to $\PP(V_*^u>v|k,\overline{\mc{C}_1})$ as $\PP(V_*^u>v|m,\overline{\mc{C}_1})$ via the substitution $M := n-K$ to simplify expressions. As shown in the appendix (see supplementary material), evaluation of the integral gives
\bea
\PP(V_*^u>v|m,\overline{\mc{C}_1}) =
\left \{ \begin{array}{cc}
\PP(V_*^u>v|m,\overline{\mc{C}_1})_{\le \frac{1}{2}} & v \le \frac{1}{2} \\
\PP(V_*^u>v|m,\overline{\mc{C}_1})_{>\frac{1}{2}} & v > \frac{1}{2},
\end{array} \right.
\eea
where
\bea
\PP(V_*^u>v|m,\overline{\mc{C}_1})_{\le \frac{1}{2}} &=& \frac{1}{3(3+m)} \left (  2m(1-2 v)^m+m (1-v)^m+ 9(1-v)^{3+m} \right. \n \\
&& -12m(1-2 v)^m v -3m (1-v)^m v + 24m(1-2 v)^m v^2 \n \\
&&\left. +3m(1-v)^m v^2-16 m (1-2 v)^m v^3 -m (1-v)^m v^3 \right ), \\
 \PP(V_*^u>v|m,\overline{\mc{C}_1})_{>\frac{1}{2}}&=& \frac{1}{3} (1-v)^{3+m}+\frac{2 (1-v)^{3+m}}{3+m}.
\eea
Calculating the expected value gives a surprisingly simple expression
\be
\E[V_*^u|m,\overline{\mc{C}_1}] = \int_0^1 \PP(V_*^u>v|m,\overline{\mc{C}_1}) \dd v  = \frac{9+2 m}{3(3+m)(4+m)}.
\ee
We now consider the case $\mc{C}_1$ where the first item is critical. By Lemma \ref{c1}, each $V_j$ for $j \ge 2$ is a deterministic function of $G$ and $W_j$. All $W_j$ for $j \ge 2$ are independent by Lemma \ref{othersindep}, so
\bea
\PP(V_* > v | g, \mc{C}_1) =\prod_{j=2}^n \PP(V_j > v | g, \mc{C}_1) = (1-v)^{(n-1)} \I(v<g).
\eea
Integrating over $G$ by Theorem \ref{greedythm}, we have
\bea
\PP(V_*>v|\mc{C}_1) = \int_0^1 \PP(V_* > v | g, \mc{C}_1) f_G(g) \dd g = (1-v)^{(n-1)}(1-2v+v^2),
\eea
which can be used to calculate the expected value. Finally, accounting for all cases of $K$ using total expectation and Lemma \ref{lemmakunif},
\bea
\E[V_*] &\le& \frac{1}{n} \E[V_*|\mc{C}_1] + \frac{1}{n} \sum_{m=0}^{n-2} \E[V_u^*|\overline{\mc{C}_1}, m]  = \frac{1}{n(2+n)} + \frac{1}{n} \sum_{m=0}^{n-2}  \frac{9+2 m}{3(3+m)(4+m)}. \label{bndsumupper}
\eea
Throughout all of the analysis in this section, we have implicitly assumed that $\sum_{i=1}^nW_i>B$. Making this condition explicit gives the desired bound. \qed

\section{Conclusion}
\label{conclusion}
We have shown strong performance bounds for both the consecutive rollout and exhaustive rollout techniques on the subset sum problem and knapsack problem. These results hold after only a single iteration and provide bounds for additional iterations.  Simulation results indicate that these bounds are very close in comparison with realized performance of a single iteration. We presented results characterizing the asymptotic behavior (asymptotic with respect to the total number of items) of the expected performance of both rollout techniques for the two problems.

An interesting direction in future work is to consider a second iteration of the rollout algorithm. The worst-case analysis of rollout algorithms for the knapsack problem in \cite{bertazzi12} shows that running one iteration results in a notable improvement, but it is not possible to guarantee additional improvement with more iterations for the given base policy. This behavior is generally not observed in practice \cite{bertsekas99}, and is not a limitation in the average-case scenario. 
A related topic is to still consider only the first iteration of the rollout algorithm, but with a larger lookahead length (e.g. trying all pairs of items for the exhaustive rollout, rather than just each item individually). Finally, it is desirable to have theoretical results for more complex problems. Studying problems with multidimensional state space is appealing since these are the types of problems where rollout techniques are often used and perform well in practice. In this direction, it would be useful to consider problems such as the bin packing problem, the multiple knapsack problem, and the multidimensional knapsack problem.
\bibliographystyle{jota}
\bibliography{average_case_rollout_chop_kp_jota_rev1_arxiv}

\end{document}


\newtheoremstyle{dotless}{}{}{\itshape}{}{\bfseries}{}{ }{}
\theoremstyle{dotless}
\newtheorem{definition}{Definition}
\newtheorem{coro}{Corollary}[section]
\newtheorem{conj}{Conjecture}
\newtheorem{thm}{Theorem}[section]
\newtheorem{prf}{Proof}
\newtheorem{lma}{Lemma}[section]
\newtheorem{prp}{Proposition}
\newcommand{\R}{\mathbb{R}}
\newcommand{\E}{\mathbb{E}}
\newcommand{\PP}{\mathbb{P}}
\newcommand{\I}{\mathbb{I}}
\newcommand{\lam}{\lambda}
\newcommand{\bs}{\boldsymbol}
\newcommand{\ph}{{\hat \pi}}
\newcommand{\xth}{\hat{x}_t}
\newcommand{\be}{\begin{equation}}
\newcommand{\ee}{\end{equation}}
\newcommand{\bea}{\begin{eqnarray}}
\newcommand{\eea}{\end{eqnarray}}
\newcommand{\bfl}{\begin{flalign}}
\newcommand{\efl}{\end{flalign}}
\newcommand{\bfc}{\begin{figure}\begin{center}}
\newcommand{\efc}{\end{center}\end{figure}}
\newcommand{\n}{\nonumber}
\newcommand{\dd}{\mathrm{d}}
\newcommand{\nin}{\noindent}
\newcommand{\mc}{\mathcal}
\newcommand{\ds}{\displaystyle}
\newcommand{\la}{\leftarrow}
\newcommand{\uga}{\underline{ga}}
\newcommand{\oga}{\overline{ga}}

\title{ Average-Case Performance of Rollout Algorithms \\ for Knapsack Problems: Supplementary Material}
\author{%
Andrew Mastin\thanks{Department of Electrical Engineering and Computer Science, Laboratory for Information and Decision Systems, Massachusetts Institute of Technology, Cambridge, MA 02139. Corresponding author.  \texttt{mastin@mit.edu}}%
\and
Patrick Jaillet\thanks{Department of Electrical Engineering and Computer Science
, Laboratory for Information and Decision Systems,
Operations Research Center, Massachusetts Institute of Technology, Cambridge, MA 02139; \texttt{jaillet@mit.edu}}%
}
\date{}
\maketitle

This document contains an index of symbols (Section \ref{slist}), proofs that were omitted in the main paper (Sections \ref{sconsec} and \ref{sexh}), and an appendix (Section \ref{apdx}) showing evaluations of integrals. 
\section{List of symbols}
\label{slist}

\begin{tabular}{ll}
$\leftarrow$ & assignment \\
$ := $ & definition \\
$ \vee $ & disjunction \\
$(\cdot)_+$ & positive part \\
$A$ & weight of the last packed item, $A:=W_{K-1}$ \\
$B$ & knapsack capacity \\
$\mc{C}_j$ & event that item $j$ is the critical item \\
$\bar{\mc{C}_1}$ & event that the first item is not critical \\
$\bar{\mc{C}_{1n}}$ & event that neither the first nor the last item are critical \\
$\mc{D}_j$ & event that item $j$ is the drop critical item \\
$\mc{E}$ & indicates that $g+w_{K-1} < 1$ \\
$\E[\cdot]$ & expectation \\
$f_X(x)$ & probability density function for the random variable $X$ \\
$G$ & gap given by the \textsc{Blind-Greedy} algorithm \\
$H(\cdot)$ & harmonic number \\
$\I(\cdot)$ & indicator function \\
$ \bs{I}$ & index sequence, $\bs{I} := \langle 1,2,\ldots,n \rangle$ \\
$K$ & critical item index (first item that is infeasible for \textsc{Blind-Greedy} to pack) \\
$L_1$ & drop critical item index (first infeasible item for \textsc{Blind-Greedy} when the first item is skipped) \\
$L_j,~j \ge 2$ & insertion critical item index (first infeasible item for \textsc{Blind-Greedy} when $j$ is inserted first) \\
$M$ & number of remaining items, $M:= n-K$ \\
$n$ & number of items \\
$O(\cdot)$ & asymptotic growth \\
$\PP(\cdot)$ & probability \\
$P_i$ & profit of item $i$ (for the knapsack problem) \\
$\bs{P_S}$ & sequence of item profits, indexed by sequence $\bs{S}$ \\
$Q$ & profit of last packed item, $Q:= P_{K-1}$ \\
$\R^+$ & positive real numbers \\
$\bs{S}$ & index sequence \\
$\mc{U}[x,y]$: & uniform distribution with support $[x,y]$ \\
\end{tabular}
\begin{tabular}{ll}
$V_1$ & drop gap (gap when the first item is skipped) \\
$V_j,~j \ge 2$ & insertion gap (gap when item $j$ is inserted first) \\
$V_j^u$ & upper bound on $V_j$ \\
$V_*$ & minimum gap; gap obtained after the first iteration of the rollout algorithm \\
$V_*^u$ & upper bound on $V_*^u$ \\
$W_i$ & weight of item $i$ \\
$\bs{W_S}$ & sequence of item weights, indexed by sequence $\bs{S}$ \\
$Z_j,~j \ge 2$ & insertion gain (gain when item $j$ is inserted first) \\
$Z_j^l$ & lower bound on $Z_j$ \\
$Z_*$ & maximum gain; gain obtained after the first iteration of the rollout algorithm \\ 
$Z_*^l$ & lower bound on $Z_*$ \\
\end{tabular} 

\section{Consecutive rollout}
\label{sconsec}
\subsection{Consecutive rollout: subset sum problem analysis}
\nin\textit{Proof of Lemma \ref{ssccn}.}
We follow the same approach that we used for Lemma \ref{sscbc1n}. Figure \ref{C3} shows the drop gap $V_1$ as a function of $w_1$ given $w_n$ and $g$. The figure is justified using the same arguments that are in the proof of Lemma \ref{sscbc1n}, but since no other items can become feasible, we can derive an exact expression for the probability of the event $V_1>v$ when only conditioning on $(g,w_n)$. Since $W_1$ has distribution $\mc{U}[0,1]$ via Lemma \ref{exhindep}, we can simply take the total length of the bold regions to find $\PP(V_1>v|\mc{C}_n, w_n, g)$.  Thus,
\be
\PP(V_1>v|\mc{C}_n, w_n, g) = (w_n-g)+(1-w_n+g-v) = (1-v), \qquad v < g,
\ee
where we have that $1-w_n+g-v$ is nonnegative since $v < g$ and $w_n \le 1$.
\begin{figure}[h]
\begin{center}
\includegraphics[scale=0.7] {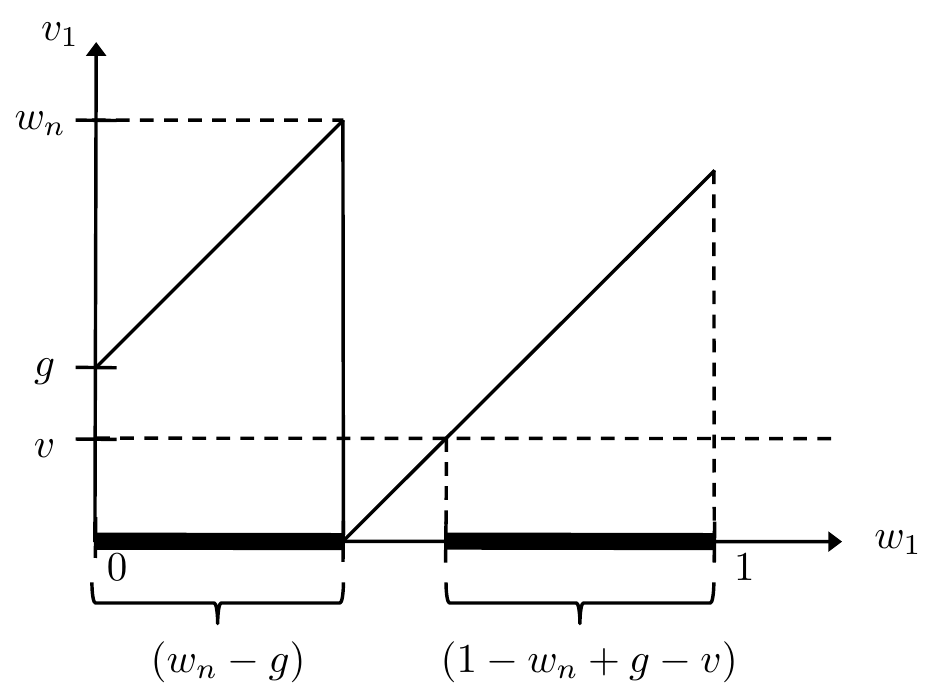}
\caption{Drop gap $v_1$ as a function of $w_1$, parameterized by $(w_n,g)$, resulting from the removal of the first item and assuming that the last item is critical ($K=n$). The function starts at $g$ and increases at unit rate until $w_1=w_n-g$, where it drops to zero, and then continues to increase at unit rate. If we only condition on ($w_n$, $g$), the probability of the event $V_1>v$ is given by the total length of the bold regions for $v < g$.  \label{C3}}
\end{center}
\end{figure}
To find the probability of the event $V_* > v$, we note that the events $V>v$ and $G>v$ are conditionally independent given $G=g$, so
\bea
\PP(V>v, G>v|\mc{C}_n,w_n,g) &=& (1-v)\I(v < g),
\eea
Marginalizing over $G$ using Lemma \ref{greedygapcond} gives
\bea
\PP(V>v, G>v|\mc{C}_n,w_n) &=& \int_0^1 \PP(V>v,G>v|\mc{C}_n,w_n,g)f_{G|\mc{C}_n,W_n}(g|\mc{C}_n,w_n)\dd g \n \\
&=& \frac{(w_n-v)(1-v)}{w_n}.
\eea
Noting the distribution of the critical item from Lemma \ref{othersindep},
\bea
\PP(V>v, G>v|\mc{C}_n) &=& \int_0^1 \PP(V>v, G>v|\mc{C}_n,w_n) f_{W_n|\mc{C}_n}(w_n|\mc{C}_n) \dd w_n \n \\
&=& 1-3v+3v^2-v^3 =\PP(V^* > v | \mc{C}_n).
\eea
This is sufficient for calculating the expected value. \qed
~\\ \\
\nin \textit{Proof of Lemma \ref{sscc1}.}~~
We use a more direct approach when the first item is critical since $W_1$ no longer has a uniform distribution (from Lemma \ref{othersindep}).  However, the analysis here is similar to the proof of Lemma \ref{sscbc1n} in how we bound the drop gap. Note that we have $B=G$ for this case. Additionally, the gap given by the minimum gap will always be equal to the drop gap since \textsc{Blind-Greedy} does not pack any items. We define a variable $V_1^u$ that satisfies $V_1 \le V_1^u$ with probability one, where
\bea
V_1^u := \left \{
\begin{array}{ll}
V_1 & L_1 = 2 \vee L_1 = 3 \\
G -W_2 - W_3 & L_1 \ge 4.
\end{array}
\right.
\eea
We let the event $L_1 \ge 4$ also account for the case where $n = 3$ and the two remaining items are feasible. If in fact $n \ge 4$ and $L_1 \ge 4$, then $V_1^u$ does not account for the additional reductions in the gap caused by more items becoming feasible. Thus we see that $V_1^u$ is a deterministic function of $(G,W_2,W_3)$.

To further simplify our expressions, we define $\mc{D}_2$, $\mc{D}_3$, $\mc{D}_{4+}$ to be the events $L_1 = 2$, $L_1=3$, and $L_1 \ge 4$, respectively. Based on these cases, the drop gap bound $V_1^u$ is given by the values shown in Table \ref{c0table}.
\begin{table}[h]
\begin{center}
\caption{Drop gap bound values when the first item is critical ($\mc{C}_1$). \label{c0table}}
\begin{tabular}{|c|c|c|}
\hline
Case & Defining inequalities & Minimum gap bound \\
\hline
$\mc{D}_2$ & $W_2 > G$ & $V_1^u = G$ \\
$\mc{D}_3$ & $W_2 \le G, W_2+W_3>G$ & $V_1^u = G-W_2$ \\
$\mc{D}_{4+}$ & $W_2+W_3 \le G$ & $V_1^u = G-W_2-W_3$ \\
\hline
\end{tabular}
\end{center}
\end{table}

We begin by finding some necessary distributions for the cases. For case $\mc{D}_3$, the posterior distribution of $W_2$ is needed. We have
\bea
f_{W_2|\mc{C}_1,\mc{D}_3,G}(w_2 | \mc{C}_1, \mc{D}_3,g) &=& \frac{\PP(W_2\le G, W_2+W_3 > G | \mc{C}_1,g,w_2)f_{W_2}(w_2)}{\PP(W_2 \le G, W_2+W_3>G|\mc{C}_1,g)},
\eea
where we have used Bayes' theorem and that $f_{W_2|\mc{C}_1,G}(w_2,\mc{C}_1,g)=f_{W_2}(w_2) = \mc{U}[0,1]$ by Lemma 3.4.
For the numerator, we have
\be
\PP(W_2\le G, W_2+W_3 > G | \mc{C}_1, g, w_2) = (1-g+w_2) \I (w_2 \le g),
\ee
which follows using Lemma 3.4 for the distribution of $W_3$.
Integrating over $W_2$ gives
\bea
\PP(\mc{D}_3|\mc{C}_1,g) &=& \int_0^1 \PP(W_2\le g, W_2+W_3 > G |\mc{C}_1,w_2) f_{W_2}(w_2) \dd w_2 = g-\frac{g^2}{2}.
\eea
Returning to the posterior distribution of $W_2$,
\bea
f_{W_2|\mc{C}_1,\mc{D}_3,G}(w_2 | \mc{C}_1, \mc{D}_3,g) &=& \frac{2(1-g+w_2)}{(2-g)g},~~ 0 \le w_2 \le g, \\
\PP(W_2 \le w_2 | \mc{C}_1,\mc{D}_3,g) &=& \frac{(2-2g+w_2)w_2}{(2-g)g},~~0 \le w_2 \le g.
\eea
Moving to the case $\mc{D}_{4+}$, define $W' := W_2+W_3$;
\bea
\PP(\mc{D}_{4+}|\mc{C}_1,g) = \PP(W' \le g|\mc{C}_1,g) = \frac{g^2}{2},
\eea
where we have used that the distribution of $W'$ conditioned on the first item being critical is the distribution for the sum of two independent uniform random variables (via Lemma 3.4). Finally for the posterior distribution of $W_2+W_3$, we have
\be
\PP(W' \le w'|\mc{C}_1,\mc{D}_{4+},g) = \frac{w'^2}{g^2},~~0 \le w' \le g.
\ee

We can now find distributions for $V_1^u$ conditioned on all cases for the drop critical item. For case $\mc{D}_2$, it is clear that $V_1^u = G$, and 
\be
\PP(\mc{D}_2|\mc{C}_1,g) = \PP(W_2>g) = 1-g.
\ee
 For $\mc{D}_3$, we have
\bea
\PP(V_1^u > v |\mc{C}_1,\mc{D}_3,g) &=& \PP(W_2 < G-v|\mc{C}_1,\mc{D}_3,g) \n \\\
&=& \frac{(2-g-v)(g-v)}{(2-g)g},~~0 \le v < g.
\eea
Then for $\mc{D}_{4+}$,
\bea
\PP(V_1^u > v |\mc{C}_1,\mc{D}_{4+},g) &=& \PP(W'  < G-v | \mc{C}_1,\mc{D}_{4+},g) \n \\
&=& \frac{(g-v)^2}{g^2},~~0 \le v < g.
\eea
Considering all three cases, we have
\bea
\PP(V_1^u > v|\mc{C}_1,g) &=& \PP(V_1^u > v|\mc{C}_1,\mc{D}_2,g) \PP(\mc{D}_2|\mc{C}_1,g) + \PP(V_1^u > v |\mc{C}_1, \mc{D}_3,g) \PP(\mc{D}_3|\mc{C}_1,g) \n \\
&& +~\PP(V_1^u > v |\mc{C}_1,\mc{D}_{4+},g) \PP(\mc{D}_{4+}|\mc{C}_1,g) \n \\
&=&(1-v-g v+v^2)\I(v < g).
\eea
This gives the expected value bound 
\be
\E[V_*|\mc{C}_1,g] \le \E[V_1^u|\mc{C}_1,g] = g - \frac{g^2}{2}-\frac{g^3}{6},
\ee
Finally, integrating over $G$ using Theorem \ref{greedythm},
\be
\E[V_*|\mc{C}_1] \le \int_0^1 \E[V_*|\mc{C}_1,g] f_G(g) \dd g = \frac{7}{30}.
\ee
\qed

~\\
\subsection{Consecutive rollout: knapsack problem analysis}
The analysis of the consecutive rollout algorithm for the knapsack problem follows the same structure as the analysis for the subset sum problem, and makes use of the properties described in Section \ref{modelgreedy}. The development here assumes that the reader is familiar with the subset sum analysis, so less detail is presented. 

We use the same definition of the drop critical item $L_1$ that was used on the subset sum problem. From the algorithm description of \textsc{Consecutive-Rollout} and the gain definition in \eqref{gaindef}, we have that the gain $Z_*(n)$ satisfies
\bea
Z_*(n) = \max \left(0,\sum_{i = 2}^{L_1-1} P_i - \sum_{i=1}^{K-1} P_i\right).
\eea
We will sometimes write $Z_*(n)$ simply as $Z_*$. The following three lemmas bound the gain $Z_*(n)$ for different cases of the critical item assuming $n \ge 3$. Theorem \ref{kp_consec_thm} then follows easily. We implicitly assume that $\sum_{i=1}^n W_i > B$ holds for the section.
\begin{lma}
For $K = n$, the expected gain satisfies
\be
\E[Z_*(n)|K=n] = \frac{1}{9}.
\ee
\end{lma}
\begin{proof}
A positive gain can only obtained in the case where the last item becomes feasible when removing the first. Consistent with our subset sum notation, let $\mc{D}_{n+1}$ denote the event that item $n$ becomes feasible when the first item is removed. Using Lemma \ref{exhindep} and the perspective of Figure \ref{vis}, this probability is given by 
\be
\PP(\mc{D}_{n+1} |g, w_n,  \mc{C}_n) = \PP(W_1 \ge W_n-G|w_n,g,\mc{C}_n) = (1-w_n+g).
\ee
Integrating over $G$ using Lemma \ref{greedygapcond} and $W_n$ using Lemma \ref{exhindep} gives
\bea
\PP(\mc{D}_{n+1}|\mc{C}_n, w_n) = \int_0^{w_n} (1-w_n+g)\frac{1}{w_n} \dd g = 1-\frac{w_n}{2},
\eea
\bea
\PP(\mc{D}_{n+1}|\mc{C}_n) = \int_0^1 \left (1-\frac{w_n}{2} \right )2w_n \dd w_n= \frac{2}{3}.
\eea
Now assuming that item $n$ becomes feasible, we are interested in the case where it provides a larger value. This is simply given by the probability
\be
\PP(P_n \ge P_1) = \frac{1}{2},
\ee
following from the symmetry of the distributions of $P_1$ and $P_n$. Conditioned on the event $P_n \ge P_1$, we are interested in the distribution of the gain, which is equal to $P_n - P_1$. We have
\be
\PP(P_n-P_1 \le q | P_n \ge P_1) = \frac{\PP(0 \le P_n-P_1 \le q)}{\PP(P_n\ge P_1)}.
\ee
For the numerator,
\bea
\PP(0 \le P_n-P_1 \le q) &=& \int_0^{1-q} \int_{p_1}^{p_1+q} \dd p_n \dd p_1 + \int_{1-q}^1 \int_{p_1}^1 \dd p_n \dd p_1 =q-\frac{q^2}{2},
\eea
which gives
\be
\PP(P_n-P_1 \le q | P_n \ge P_1) =2q- q^2,
\ee
\be
\E[P_n-P_1| P_n \ge P_1] = \frac{1}{3}.
\ee
Finally, by the independence of item weight and profit, we have
\be
\E[Z_*(n)| \mc{C}_n] = \E[P_n-P_1 | P_n \ge P_1] \PP(P_n \ge P_1) \PP(\mc{D}_{n+1} | \mc{C}_n) = \frac{1}{3} \cdot \frac{1}{2} \cdot \frac{2}{3} = \frac{1}{9}.
\ee
\end{proof}

\begin{lma}
For $2 \le K \le n-1$, the expected gain satisfies
\be
\E[Z_*(n)|2 \le K \le n] \ge \frac{59}{288} \approx 0.205.
\ee
\end{lma}
\begin{proof}
We again let $\bar{\mc{C}_{1n}}$ be the event that $2 \le K \le n-1$. We fix $K=k$, and the proof holds for all valid values of $k$. In the case of event $\bar{\mc{C}_{1n}}$, it is possible that removing the first item allows for the critical item to become feasible as well as additional items (i.e. $L_1 \ge k+2$). However, we are only guaranteed the existence of one item beyond the critical item since it is possible that $k=n-1$. Let $\mc{D}_{k+1}$ indicate the event $L_1 = k+1$ and let $\mc{D}_{(k+2)+}$ indicate the event $L_1 \ge k+2$. If item $k+2$ does not exist (i.e. $k = n-1$), then this event means that all remaining items are packed.

For the probability of the event  $\mc{D}_{k+1}$, we have from Lemma \ref{exhindep} that $W_1$ has distribution $\mc{U}[0,1]$. Then,
\be
\PP(\mc{D}_{k+1}|g, w_k, w_{k+1}, \bar{\mc{C}_{1n}}) = w_{k+1} - (w_k-g+w_{k+1}-1)_+.
\ee
This can be argued using an illustration similar to Figure \ref{C2}, where the second term mitigates that case where $w_{k+1}$ extends beyond $b+1$ for $B=b$.
Likewise, for the event $\mc{D}_{(k+2)+}$, we have
\be
\PP(\mc{D}_{(k+2)+}|g, w_k, w_{k+1},\bar{\mc{C}_{1n}}) = (1-w_k+g-w_{k+1})_+.
\ee
Starting with event $\mc{D}_{k+1}$, we integrate over $W_{k+1}$, which has uniform density by Lemma \ref{exhindep}.
\bea
\PP(\mc{D}_{k+1} |g,w_k,\bar{\mc{C}_{1n}}) &=& \int_0^1 \PP(\mc{D}_{k+1}| g, w_k, w_{k+1},\bar{\mc{C}_{1n}})f_{W_{k+1}}(w_{k+1}) \dd w_{k+1} \n \\
&=& \int_0^1 w_{k+1} \dd w_{k+1} - \int_{1+g-w_k}^1(w_k-g+w_{k+1}-1) \dd w_{k+1} \n \\
&=& \frac{1}{2} - \frac{g^2}{2} + g w_k - \frac{w_k^2}{2}.
\eea
Marginalizing over $G$ with Lemma \ref{greedygapcond} gives
\bea
\PP(\mc{D}_{k+1} |w_k, \bar{\mc{C}_{1n}}) &=& \int_0^1 \PP(\mc{D}_{k+1} |g,w_k,\bar{\mc{C}_{1n}}) f_{G|W_k}(g|w_k) \dd g \n \\
&=& \int_0^{w_k} \left (\frac{1}{2} - \frac{g^2}{2} + g w_k - \frac{w_k^2}{2} \right) \frac{1}{w_k} \dd g \n \\
&=& \frac{1}{2} - \frac{w_k^2}{6}.
\eea
Finally by Lemma \ref{othersindep},
\bea
\PP(\mc{D}_{k+1}|\bar{\mc{C}_{1n}}) = \int_0^1 \PP(\mc{D}_{k+1} |w_k,\bar{\mc{C}_{1n}}) f_{W_k}(w_k) \dd w_k = \int_0^1 \left ( \frac{1}{2} - \frac{w_k^2}{6} \right ) 2w_k \dd w_k = \frac{5}{12}.
\eea

Now for the event $\mc{D}_{(k+2)+}$, we integrate in the same order, using the same lemmas.
\bea
\PP(\mc{D}_{(k+2)+} |g, w_k,\bar{\mc{C}_{1n}}) &=& \int_0^1 \PP(\mc{D}_{(k+2)+}|g,w_k,w_{k+1},\bar{\mc{C}_{1n}})f_{W_{k+1}}(w_{k+1}) \dd w_{k+1} \n \\
&=& \int_0^{1-w_k+g}(1-w_k+g-w_{k+1}) \dd w_{k+1} \n \\
&=& \frac{1}{2} + g + \frac{g^2}{2} - w_k - g w_k + \frac{w_k^2}{2}.
\eea
\bea
\PP(\mc{D}_{(k+2)+}|w_k,\bar{\mc{C}_{1n}}) &=& \int_0^1 \PP(\mc{D}_{(k+2)+}|g,w_k,\bar{\mc{C}_{1n}}) f_{G|W_k}(g | w_k) \dd g \n \\
&=& \int_0^{w_k}\left ( \frac{1}{2} + g + \frac{g^2}{2} - w_k - g w_k + \frac{w_k^2}{2} \right )\frac{1}{w_k} \dd g \n \\
&=& \frac{1}{2} - \frac{w_k}{2} + \frac{w_k^2}{6}.
\eea
\bea
\PP(\mc{D}_{(k+2)+}|\bar{\mc{C}_{1n}}) &=& \int_0^1 \PP(\mc{D}_{(k+2)+}|\bar{\mc{C}_{1n}},w_k) f_{W_k}(w_k) \dd w_k = \int_0^1 \left ( \frac{1}{2} - \frac{w_k}{2} + \frac{w_k^2}{6} \right ) 2 w_k \dd w_k = \frac{1}{4}.
\eea

Equipped with these probabilities, we now consider the gain from the rollout for the different drop critical item cases. For the case where only one item becomes feasible ($\mc{D}_{k+1}$), the analysis in the previous lemma holds, so we have
\be
\E[P_n-P_1|\bar{\mc{C}_{1n}}, \mc{D}_{k+1}] = \E[P_n-P_1|P_n>P_1] \PP(P_n>P_1) \PP(\mc{D}_{k+1}|\bar{\mc{C}_{1n}}) = \frac{1}{3} \cdot \frac{1}{2} \cdot \frac{5}{12} = \frac{5}{72}.
\ee
If two or more items become feasible ($\mc{D}_{(k+1)+}$), we only consider the gain resulting from adding two items, and this serves as a lower bound for the case of more items becoming feasible. Accordingly, define
\be
P' := P_k + P_{k+1}.
\ee
The probability that the profits of the two items are greater than $P_1$ is given by
\bea
\PP(P' \ge P_1) = 1 - \PP(P' < P_1) = 1 - \int_0^1 \int_0^{p_1} p' \dd p' \dd p_1 = 1 - \int_0^1 \frac{p^2_1}{2} \dd p_1 = \frac{5}{6}.
\eea
The gain conditioned on the event $P' > P_1$ is given by
\bea
\PP(P' - P_1 \le q | P' \ge P_1) = \frac{\PP(0 \le P' - P_1 \le q)}{\PP(P' \ge P_1)}.
\eea
Proceeding with the numerator and assuming $0 \le q \le 1$,
\bea
\PP(0 \le P'-P_1 \le q) &=& \int_0^{1-q} \int_{p_1}^{p_1+q} p' \dd p' \dd p_1 + \int_{1-q}^1 \int_{p_1}^{1}p' \dd p' \dd p_1 + \int_{1-q}^1 \int_1^{p_1+q}(2-p') \dd p' \dd p_1 \n \\
&=& \int_0^{1-q} \left ( p_1 q + \frac{q^2}{2} \right ) \dd p_1 + \int_{1-q}^1 \left ( \frac{1}{2} - \frac{p_1^2}{2} \right ) \dd p_1 \n \\
&& + \int_{1-q}^1 \left ( -\frac{3}{2} + 2p_1 - \frac{p_1^2}{2} + 2q - p_1 q - \frac{q^2}{2} \right ) \dd p_1 \n \\
&=& \frac{q}{2} + \frac{q^2}{2} - \frac{q^3}{3},~ 0 \le q \le 1.
\eea
Now for $1 < q \le 2$,
\bea
\PP(0 \le P'-P_1 \le q) &=& \int_0^1 \int_{p_1}^1 p' \dd p' \dd p_1 + \int_{0}^{2-q} \int_1^{p_1+q}(2-p') \dd p' \dd p_1 + \int_{2-q}^1 \int_1^2(2-p') \dd p' \dd p_1 \n \\
&=& \int_0^1 \left ( \frac{1}{2} - \frac{p_1^2}{2} \right ) \dd p_1 + \int_0^{2-q} \left ( -\frac{3}{2} + 2p_1 - \frac{p_1^2}{2} + 2q - p_1 q - \frac{q^2}{2}  \right ) \dd p_1 \n \\
&&+ \frac{1}{2} \int_{2-q}^1 \dd p_1 \n \\
&=& -\frac{1}{2} + 2q - q^2 +\frac{q^3}{6},~1 < q \le 2.
\eea
The distribution for the gain is thus given by
\be
\PP(P'-P_1\le q | P'-P_1 \ge 0) = \left \{ \begin{array}{cc}
\frac{3}{5}q + \frac{3}{5}q^2 - \frac{2}{5} q^3 & 0\le q \le 1 \\
-\frac{3}{5} + \frac{12}{5} q - \frac{6}{5} q^2 + \frac{1}{5}q^3 & 1 < q \le 2. \\
\end{array}\right.
\ee
The expected value is
\be
\E[P'-P_1 | P'-P_1 \ge 0] = \int_0^1 q \left(\frac{3}{5}+\frac{6}{5}q-\frac{6}{5}q^2 \right) \dd q + \int_1^2 q \left ( \frac{12}{5} - \frac{12}{5}q + \frac{3}{5}q^2 \right ) \dd q = \frac{13}{20}.
\ee
Recalling that it is possible for more than two items to be added in the case $\mc{D}_{(k+2)+}$, let $P''$ be the total value of items added for the case. We may bound the expected gain as follows, where the term $\PP(\mc{D}_k|\bar{\mc{C}_{1n}})$ is omitted since it provides zero gain. We are implicitly using the fact that item weights and profits are independent.
\bea
\E[Z_*(n)|\bar{\mc{C}_{1n}}] &=& \E[P''-P_1|P'' > P_1] \PP(P'' \ge P_1) \PP(\mc{D}_{(k+2)+}|\bar{\mc{C}_{1n}}) \n \\
 && + \E[P_n - P_1 | P_n \ge P_1] \PP(P_n \ge P_1) \PP(\mc{D}_{k+1}|\bar{\mc{C}_{1n}}) \n \\
&\ge& \E[P'-P_1|P' > P_1] \PP(P' \ge P_1) \PP(\mc{D}_{(k+2)+}|\bar{\mc{C}_{1n}}) \n \\
&& + \E[P_n - P_1 | P_n \ge P_1] \PP(P_n \ge P_1) \PP(\mc{D}_{k+1}|\bar{\mc{C}_{1n}}) \n \\
&=& \frac{13}{20} \cdot \frac{5}{6} \cdot \frac{1}{4} + \frac{1}{3} \cdot \frac{1}{2} \cdot  \frac{5}{12} = \frac{59}{288}.
\eea
\end{proof}

\begin{lma}
For $K = 1$, the expected gain satisfies
\be
\E[Z_*(n)|K=1] \ge \frac{5}{24} \approx 0.208.
\ee
\end{lma}
\begin{proof}
We use the drop events $\mc{D}_2$, $\mc{D}_3$, and $\mc{D}_{4+}$ just as we did for the subset sum problem. The event probabilities given $G=g$ are the same as those for the subset sum problem. Accordingly,
\be
\PP(\mc{D}_2 | \mc{C}_1) = \int_0^1 \PP(\mc{D}_2 | \mc{C}_1,g) f_G(g) \dd g = \int_0^1 (1-g)(2-2g) \dd g = \frac{2}{3},
\ee
\be
\PP(\mc{D}_3 | \mc{C}_1) = \int_0^1 \PP(\mc{D}_3 | \mc{C}_1,g) f_G(g) \dd g = \int_0^1 \left ( g-\frac{g^2}{2} \right )(2-2g) \dd g = \frac{1}{4},
\ee
\be
\PP(\mc{D}_{4+} | \mc{C}_1) = \frac{1}{12}.
\ee
The greedy solution gives zero value, so the expected gain is easily determined using independence of item weights and profits,
\be
\E[Z_*|\mc{C}_1,\mc{D}_2] = 0,
\ee
\be
\E[Z_*|\mc{C}_1,\mc{D}_3] = \E[P_2] = \frac{1}{2},
\ee
\be
\E[Z_*|\mc{C}_1, \mc{D}_{4+}] \ge \E[P_2 + P_3] = 1.
\ee
Combining all cases for the drop critical item,
\bea
\E[Z_*|\mc{C}_1] &=& \E[Z_*|\mc{C}_1, \mc{D}_3] \PP(\mc{D}_3 | \mc{C}_1) + \E[Z_*|\mc{C}_1, \mc{D}_{4+}] \PP(\mc{D}_{4+}|\mc{C}_1) \n \\
&\ge& \frac{1}{2} \cdot \frac{1}{4} + 1 \cdot \frac{1}{12} = \frac{5}{24}.
\eea

\end{proof}
The result for the knapsack problem then follows.\\
~\\
\nin\textit{Proof of Theorem \ref{kp_consec_thm}.}
The events $\mc{C}_1$, $\bar{\mc{C}_{1n}}$, and $\mc{C}_n$ form a partition of the event $\sum_{i = 1}^n W_i > B$, so using Lemma \ref{lemmakunif} gives
\bea
\E \left [ Z_*(n) \left \vert \sum_{i=1}^n W_i>B \right. \right ] &=& \E[Z_*(n)|\mc{C}_1] \PP(\mc{C}_1)+\E[Z_*(n)|\bar{\mc{C}_{1n}}] \PP(\bar{\mc{C}_{1n}})+\E[Z_*(n)|\mc{C}_n] \PP(\mc{C}_n) \n \\
&\ge& \frac{5}{24}\left(\frac{1}{n}\right) + \frac{59}{288}\left(\frac{n-2}{n}\right)+\frac{1}{9}\left(\frac{1}{n}\right) = \frac{-26+59n}{288n}.
\eea
\qed

\section{Exhaustive rollout}
\label{sexh}
\subsection{Exhaustive rollout: subset sum problem analysis}
\nin \textit{Proof of Lemma \ref{crit insert}.} 
Again fix $K=k$ for $k>1$, $G=g$, $W_{k-1}=w_{k-1}$. Define the random variable $V_k^u$ so that 
\bea
V_k^u = \left \{
\begin{array}{ll}
V_k & L_k = k-1 \n \\
1 & L_k \le k-2 \vee L_k = k.
\end{array}
\right.
\eea
From Lemma \ref{exhindep} we are guaranteed that given $G=g$, $W_k$ follows distribution $\mc{U}[g,1].$ Thus to determine $\PP(V_k^u > v|g,w_{k-1}, \overline{\mc{C}}_1)$, we can use the same analysis for Lemma \ref{remain insert} but restricted to the interval $g \le w_k \le 1$. Taking the expression $\PP(V^u > v|g,w_{K-1}, \overline{\mc{C}}_1)$ in \eqref{torri}, removing the $(g-v)_+$ term, and normalizing by $(1-g)$, we have
\bea
\PP(V_k^u>v|g,w_{k-1}, \overline{\mc{C}}_1) &=& \left (\frac{1}{1-g} \right) \left ( (w_{k-1}-v)_+-(g+w_{k-1}-v-1)_++(1-g-w_{k-1})_+ \right ).
\eea
This holds for all $k>1$ so we replace $k$ with $K$ in the expression. \qed
~\\
~\\
\nin\textit{Proof of Lemma \ref{c1}.} 
Fix $G = g$. Note that for $K=1$, the $jth$ insertion gap can never be greater than $g$. Keeping the analysis for Lemma \ref{remain insert} in mind and using Lemma \ref{exhindep}, we have that for $v<g$, 
\bea
\PP(V_j > v |g,\mc{C}_1) = (g-v)+(1-g),
\eea
where $(g-v)$ corresponds to the case where $w_j \in [0,g]$ and $(1-g)$ corresponds to $w_j \in (g,1]$. \qed
~\\
~\\
\textit{Proof of Corollary \ref{exhcoro} and Theorem \ref{ssasym}}. The sum terms may be bounded with integral approximations. For the upper bound, the argument of the sum is convex in $m$, so the midpoint rule provides an upper bound.
\be
\sum_{m=0}^{n-2}  \frac{9+2 m}{3(3+m)(4+m)} \le \int_{-\frac{1}{2}}^{n-\frac{3}{2}} \frac{9+2 m}{3(3+m)(4+m)} \dd m = \log \left[ \left (\frac{3+2n}{5} \right )  \left (\frac{7}{5+2n} \right )^{1/3} \right].
\ee
The asymptotic result then follows. \qed

\subsection{Exhaustive rollout: knapsack problem analysis}
We follow the same approach that was used for the subset sum problem and assume that the reader understands this analysis (and thus less detail is included here). We employ the results from Section \ref{modelgreedy} and we use the same definition for the $j$th insertion item that was used for the subset sum problem. Analogous to the $j$th insertion gap $V_j$, we define here the \textit{$j$th insertion gain} $Z_j$, where
\bea
Z_j := \max \left ( 0,  ~\I(W_j \le B) \left ( P_j + \sum_{i=1}^{L_j-1}P_i\I(i \ne j) \right) - \sum_{i = 1}^{K-1} P_i \right ).
\eea
The $j$th insertion gain is simply the positive part of the difference between the value of the solution obtained by using \textsc{Blind-Greedy} after moving item $j$ to the front of the sequence, and the value of the solution from using \textsc{Blind-Greedy} on the original input sequence. 

We will bound the expected insertion gains while conditioning on $(G,W_{K-1},P_{K-1})$. Assuming $K>1$, this is done in the following three lemmas for packed items, the critical item, and remaining items, respectively, just as we did for the subset sum problem. The lemma after these three handles the case where $K=1$. We assume that $\sum_{i=1}^n W_i > B$ throughout the section.

\begin{lma}
\label{kppack}
For $K>1$ and $j = 1,\ldots,K-1,$ the $j$th insertion gain satisfies
\bea
Z_j = 0
\eea
with probability one.
\end{lma}
\begin{proof}
By the proof of Lemma \ref{pack insert}.
\end{proof}

\begin{lma}
\label{kpremain}
For $K>1$ and $j = K+1,\ldots,n$, the $j$th insertion gain satisfies $Z_j \ge Z_j^l$ with probability one, where $Z_j^l$ is a deterministic function of $(G, W_{K-1}, W_j, P_{K-1},P_j)$ and conditioning only on $(G ,W_{K-1},P_{K-1})$ gives
\bea
\PP(Z_j^l \le z | g, w_{K-1}, p_{K-1}, \bar{\mc{C}_1}) &=& zg + \min(z+p_{K-1},1) \left ( w_{K-1}-(g+w_{K-1}-1)_+ \right ) \n \\
&& + (1-g-w_{K-1})_+ \n \\
&=:& \PP(Z^l \le z | g, w_{K-1}, p_{K-1},\bar{\mc{C}_1}).
\eea
\end{lma}
\begin{proof}
Fix $K=k$ for any $k>1$, and let the event $\bar{\mc{C}_1}$ be implicit. We define the lower bounding random variable $Z_j^l$ so that
\bea
Z_j^l = \left \{
\begin{array}{ll}
Z_j & L_j = k \vee L_j = k-1 \n \\
0 & L_j \le k-2 \vee L_j = j.
\end{array}
\right.
\eea
This means we have an exact characterization of the $j$th insertion gain when the insertion critical item is either $k$ or $k-1$, and a worst case gain of zero value in other cases. Thus it can be seen that $Z_j \ge Z_j^l$ with probability one, and $Z_j^l$ uniquely depends on the random variables $(G, W_{K-1}, W_j, P_{K-1},P_j)$. Let $\mc{D}_k$, $\mc{D}_{k-1}$, and $\mc{D}_{(k-2)-}$ indicate the events $L_j = k$, $L_j = k-1$, and $L_j \le k-2 \vee L_j = j$, respectively. Using an illustration similar to Figure \ref{eg} under the assumption that $G=g$ and $\bs{W_S} = \bs{w_s}$, we have that if we only allow $W_j$ to be random, then by Lemma \ref{exhindep},
\bea
\PP(\mc{D}_{k}|g,w_{k-1},p_{k-1}, p_j)&=&g, \\
\PP(\mc{D}_{k-1}|g,w_{k-1},p_{k-1}, p_j)&=&w_{k-1}-(g+w_{k-1}-1)_+, \\
\PP(\mc{D}_{(k-2)-}|g,w_{k-1},p_{k-1}, p_j) &=& (1-g-w_{k-1})_+.
\eea
Note that these expressions do not depend on any of the $\bs{P_S}$ values since item weights and profits are independent. For each of the above cases, we can find the probability distribution for $Z_j^l$ while allowing only $P_j$ to be random, so that
\bea
\PP(Z_j^l \le z | \mc{D}_k,g,w_{k-1},w_j,p_{k-1}) &=& \PP(P_j \le z) = z, \\
\PP(Z_j^l \le z | \mc{D}_{k-1},g,w_{k-1}, w_j, p_{k-1}) &=& \PP(P_j - P_{k-1} \le z |p_{k-1}) = \min(z+p_{k-1},1), \\
\PP(Z_j^l \le z | \mc{D}_{(k-2)-},g,w_{k-1}, w_j, p_{k-1}) &=& \PP( 0 \le z ) = 1.
\eea
Again, these expressions do not depend on any of the $\bs{W_S}$ values by item profit and weight independence. Then, combining terms and noting that the above functions do not depend on all of the conditioned parameters, we have that if we only condition on $(G,W_{k-1},P_{k-1})$,
\bea
\PP(Z_j^l \le z | g, w_{k-1}, p_{k-1}) &=& \PP(Z_j^l \le z | \mc{D}_k,g,w_{k-1},p_{k-1}) \PP(\mc{D}_{k}|g,w_{k-1},p_{k-1}) \n \\
&& +\PP(Z_j^l \le z | \mc{D}_{k-1},g,w_{k-1}, p_{k-1})  \PP(\mc{D}_{k-1}|g,w_{k-1},p_{k-1}) \n \\
&& +\PP(Z_j^l \le z | \mc{D}_{(k-2)-},g,w_{k-1}, p_{k-1})  \PP(\mc{D}_{(k-2)-}|g,w_{k-1},p_{k-1})  \n \\
&=& zg + \min(z+p_{k-1},1) \left ( w_{k-1}-(g+w_{k-1}-1)_+ \right ) \n \\
&& + (1-g-w_{k-1})_+.
\eea
The analysis holds for all $k>1$, so we replace $k$ with $K$, which yields the expression in the lemma.
\end{proof}
\begin{lma}
For $K>1$, the $K$th insertion gap satisfies $Z_K \ge Z_K^l$ with probability one, where $Z_K^l$ is a deterministic function of $(G, W_{K-1}, W_K, P_{K-1},P_K)$ and conditioning only on $(G ,W_{K-1},P_{K-1})$ gives
\bea
\PP(Z_K^l \le z | g, w_{K-1}, p_{K-1},\bar{\mc{C}_1}) &=& \frac{1}{1-g}\left (\min(z+p_{K-1},1) \left ( w_{k-1}-(g+w_{K-1}-1)_+ \right ) + (1-g-w_{K-1})_+ \right) \n \\
&=:& \PP(\tilde{Z}^l \le z | g, w_{K-1}, p_{K-1},\bar{\mc{C}_1}).
\eea
\label{kpcrit}
\end{lma}
\begin{proof}
Fix $K=k$ for $k>1$ and make the event $\bar{\mc{C}_1}$ implicit. We define the lower bound random variable $Z_k^l$ so that
\bea
Z_k^l := \left \{
\begin{array}{ll}
Z_k & L_k = k-1 \n \\
0 & L_k \le k-2 \vee L_k = k.
\end{array}
\right.
\eea
This random variable assumes a worst-case bound of zero gain if item $k-1$ becomes infeasible. By definition, we have $Z_k \ge Z_k^l$ with probability one and that $Z_k^l$ is uniquely determined by $(G, W_{k-1}, W_j, P_{k-1},P_j)$. Let $\mc{D}_{k-1}$ be the event that $L_k = k-1$ and let $\mc{D}_{(k-2)-}$ indicate the event $L_k \le k-2 \vee L_k = k$. By Lemma \ref{exhindep}, we have that for $G=g$, item $k$ has distribution $\mc{U}[g,1]$. Using the analysis in the previous lemma but restricted to the interval $[g,1]$, we have
\bea
\PP(\mc{D}_{k-1}|g,w_{k-1},p_{k-1}, p_k)&=&\frac{1}{1-g} \left ( w_{k-1}-(g+w_{k-1}-1)_+ \right ), \\
\PP(\mc{D}_{(k-2)-}|g,w_{k-1},p_{k-1}, p_k) &=& \frac{1}{1-g} (1-g-w_{k-1})_+.
\eea
By the independence of item weights and profits, the following results carry over from the proof of the previous lemma:
\bea
\PP(Z_k^l \le z | \mc{D}_{k-1},g,w_{k-1}, w_k, p_{k-1}) &=& \PP(P_k - P_{k-1} \le z |p_{k-1}) = \min(z+p_{k-1},1), \\
\PP(Z_k^l \le z | \mc{D}_{(k-2)-},g,w_{k-1}, w_k, p_{k-1}) &=& \PP( 0 \le z ) = 1.
\eea
We then have that if we only condition on $(G,W_{k-1},P_{k-1})$,
\bea
\PP(Z_k^l \le z | g, w_{k-1}, p_{k-1}) &=& \PP(Z_k^l \le z | \mc{D}_{k-1},g,w_{k-1}, p_{k-1})  \PP(\mc{D}_{k-1}|g,w_{k-1},p_{k-1}) \n \\
&& +\PP(Z_k^l \le z | \mc{D}_{(k-2)-},g,w_{k-1}, p_{k-1})  \PP(\mc{D}_{(k-2)-}|g,w_{k-1},p_{k-1})  \n \\
&=& \frac{1}{1-g}\left (\min(z+p_{k-1},1) \left ( w_{k-1}-(g+w_{k-1}-1)_+ \right ) + (1-g-w_{k-1})_+ \right). \n \\
\eea
The analysis is valid for all $k>1$, so we replace $k$ with $K$ to obtain the expression in the lemma.
\end{proof}
~\\
~\\
We now define $Z_*(n)$, which is the gain given by the first iteration of the rollout algorithm on an instance with $n$ items,
\bea
Z_*(n) := \max(Z_1,\ldots,Z_n).
\eea
For the rest of the section, we will usually refer to $Z_*(n)$ simply as $Z_*$.

\nin\textit{Proof of Theorem \ref{kpexhthm}}.
We proceed in a fashion nearly identical to the proof of Theorem \ref{exhthm}. We have that for $K=k>1$, $Z_* \ge Z_*^l$ with probability one, where
\bea
Z_*^l := \max(Z_k^l, Z_{k+1}^l, \ldots, Z_n^l).
\eea
This makes use of Lemmas \ref{kppack} - \ref{kpcrit}. By Lemmas \ref{kpremain} and \ref{kpcrit}, each $Z_j^l$ for $j \ge k$ is a deterministic function of $(G,W_{k-1},W_j,P_{k-1},P_j)$. Lemma \ref{exhindep} gives that item weights $W_j$ for $j >k$ independently follow the distribution $\mc{U}[0,1]$ and $W_k$ independently follows the distribution $\mc{U}[g,1]$. As a result, conditioning on only $(G,W_{k-1},P_{k-1})$ makes $Z_j^l$ independent for $j \ge k$, and then by the definition of the maximum,
\bea
\PP(Z_*^l \le z | g,w_{k-1},p_{k-1},k,\overline{\mc{C}_1}) &=& \PP(Z_k^l \le z|g,w_{k-1},p_{k-1},\overline{\mc{C}_1}) \prod_{j=k+1}^n \PP(Z_j^l \le z |g,w_{k-1},p_{k-1},\overline{\mc{C}_1}) \n \\
&=& \PP(\widetilde{Z}^l \le z |g,w_{k-1}, p_{k-1}, \overline{\mc{C}_1}) \left ( \PP(Z^l \le z |g,w_{k-1},p_{k-1},\overline{\mc{C}_1})  \right )^{(n-k)}. \n \\
\label{bigpi}
\eea
In the remainder of the proof, we first integrate over the conditioned variables and then consider the case $\mc{C}_1$. For the integrals, we adopt some simplified notation to make expressions more manageable. As with the subset sum problem, let $M:=n-K$. Also define
\bea
\pi_+ &:=& g, \\
\pi_0 &:=& w_{K-1}-(g+w_{k-1}-1)_+, \\
\pi_- &:=& (1-g-w_{k-1}), \\
\tilde{\pi}_0 &:=& \frac{1}{1-g} \left ( w_{k-1}-(g+w_{k-1}-1)_+ \right ), \\
\tilde{\pi}_- &:=& \frac{1}{1-g} (1-g-w_{k-1})_+.
\eea
This allows us to write \eqref{bigpi} as
\bea
\PP(Z_*^l \le z | g,w_{k-1},p_{k-1},m,\overline{\mc{C}_1}) = \left( \min(z+p_{k-1},1) \tilde{\pi}_0 + \tilde{\pi}_- \right) (z \pi_+ + \min(z+p_{k-1},1) \pi_0 + \pi_-)^m.
\eea
Integrating over $p_{k-1}$, which follows density $\mc{U}[0,1]$,
\bea
\PP(Z^l_* \le z |g,w_{k-1},m, \bar{\mc{C}_1}) &=& \int \PP(Z_*^l \le z|g,w_{k-1},p_{k-1},m,\bar{\mc{C}_1}) f_{P_{k-1}}(p_{k-1})\dd p_{k-1} \n \\
&=& \int_0^{1-z}  \left ((z+p_{k-1})\tilde{\pi}_0+\tilde{\pi}_-\right) \left (z\pi_+ + (z+p_{k-1})\pi_0 + \pi_- \right )^m \dd p_{k-1} \n \\
&&+ \int_{1-z}^1\left (\tilde{\pi}_0+\tilde{\pi}_-\right) \left (z\pi_+ + \pi_0 + \pi_-\right )^m \dd p_{k-1} \n \\
&=&(\tilde{\pi}_0+\tilde{\pi}_-)(\pi_0+\pi_-+\pi_+z)^mz + \frac{1}{(m+1)(m+2)\pi_0^2} \cdot \n \\
&& \left ( (\pi_0+P_n+\pi_+z)^{m+1}(\pi_0 \tilde{\pi}_0 (m+1) + \pi_0 \tilde{\pi}_-(m+2)-\tilde{\pi}_0\pi_- - \tilde{\pi}_0 \pi_+z) \right. \n \\
&& \left. - (\pi_-+(\pi_0+\pi_+)z)^{m+1}((2+m)\pi_0\tilde{\pi}_-  -\pi_-\tilde{\pi}_0+\tilde{\pi}_0(\pi_0+m\pi_0-\pi_+)z) \right ). \n \\
\eea
At this point it is useful to evaluate separately the cases where $g+w_{k-1} < 1$ and $g+w_{k-1} \ge 1$. Let $\mc{E}$ indicate the event that $g+w_{k-1}<1$ holds, and let $\bar{\mc{E}}$ be the complement of this event. Also define $A:= W_{K-1}$. This allows us to define
\bea
\PP(Z^l_* \le z |g,w_{k-1},m, \bar{\mc{C}_1})_\mc{E} &:=& \PP(Z^l_* \le z |g,w_{k-1},m, \bar{\mc{C}_1}) \I(g+w_{k-1} < 1), \\
\PP(Z^l_* \le z |g,w_{k-1},m, \bar{\mc{C}_1})_{\bar{\mc{E}}} &:=& \PP(Z^l_* \le z |g,w_{k-1},m, \bar{\mc{C}_1}) \I(g+w_{k-1} \ge 1),
\eea
so that
\bea
\PP(Z^l_* \le z |g,w_{k-1},m, \bar{\mc{C}_1}) = \PP(Z^l_* \le z |g,w_{k-1},m, \bar{\mc{C}_1})_\mc{E} + \PP(Z^l_* \le z |g,w_{k-1},m, \bar{\mc{C}_1})_{\bar{\mc{E}}}.
\eea
Starting with the case where $\mc{E}$ holds and substituting $A$ for $W_{k-1}$,
\bea
\PP(Z_*^l \le z |g,a,m, \bar{\mc{C}_1})_\mc{E} &=& z(1-g+gz)^m + \frac{(1-g+gz)^{m+1}(1-g+m-gm-gz)}{(1-g)a(m+1)(m+2)} \n \\
&& -\frac{ \left ( 1-g+gz + a(1-z)\right)^{m+1}\left(1-g+m-gm-gz + a(-1-m+z+mz)\right)}{(1-g)a(m+1)(m+2)}. \n \\
\eea
We now wish to compute
\bea
\PP(Z_*^l \le z | m,\bar{\mc{C}_1})_{\mc{E}} &:=& \int_0^1 \int_0^{1-g} \PP(Z_*^l \le z |g,a,m,\bar{\mc{C}_1})_\mc{E} f_A(a) f_G(g) \dd a \dd g. \label{firstkpint}
\eea
The evaluation of this integral is given in Section \ref{firstkpintsec}, which shows
\bea
\PP(Z_*^l \le z | m, \bar{\mc{C}_1})_{\mc{E}} &=& \rho_1(m,z) + \sum_{j=1}^{m+1} \rho_{2j}(m,z) + \rho_3(m,z) + \rho_4(m,z),
\eea
where
\bea
\rho_1(m,z) &=& -\frac{2 z \left(2+m^2 (-1+z)^2+m (-1+z) (-3+5 z)-2 z \left(3-3 z+z^{2+m}\right)\right)}{(1+m) (2+m) (3+m) (-1+z)^3}, \\
\rho_{2j}(m,z) &=& \frac{2z^{3+m} (j+(2+m) (-2+z)-j z)}{j (-3+j-m) (-2+j-m) (1+m) (2+m) (-1+z)^2}, \\
&& + \frac{2z^j (-j (1+m) (-1+z)+(2+m) (-1+m (-1+z)+2 z))}{j (-3+j-m) (-2+j-m) (1+m) (2+m) (-1+z)^2}, \\
\rho_3(m,z) &=& -\frac{2 H(m+1) \left(-1+m (-1+z)+2 z+(-2+z) z^{3+m}\right)}{(1+m) (2+m) (3+m) (-1+z)^2}, \\
\rho_4(m,z) &=& -\frac{2}{(2+m)^2 (3+m) (-1+z)}-\frac{2 z^{2+m}}{(2+m)^2}+\frac{2 z^{3+m}}{(2+m)^2 (3+m) (-1+z)}.
\eea
Since we are ultimately interested in the expected value of $Z_*^l$, we wish to evaluate
\be
\bar{\E}[Z_*^l|m, \bar{\mc{C}_1}]_\mc{E} := \int_0^1 \PP(Z_*^l \le z | m, \bar{\mc{C}_1})_\mc{E} \dd z.
\ee
Recall that $\bar{\E}[\cdot] := 1-\E[\cdot]$. Using the definition
\bea
\xi_j(m) := \int_0^1\rho_j(m,z)\dd z,
\eea
we have
\be
\bar{\E}[Z_*^l|m,\bar{\mc{C}_1}]_{\mc{E}} = \xi_1(m) + \sum_{j=1}^{m+1} \xi_{2j}(m) + \xi_3(m) + \xi_4(m),
\ee
where
\bea
\xi_1(m) &=& -\frac{2 H(m+1)  (3+m-H(m+3)(2+m))}{(m+1)(m+2)(m+3)}, \\
\xi_{2j}(m) &=& \frac{2 \left(-(-3+j-m) (2+m)+\left(j+(2+m)^2\right) H(j)-\left(j+(2+m)^2\right)H(m+3)\right)}{j (-3+j-m) (-2+j-m) (1+m) (2+m)}, \\
\xi_3(m) &=& \frac{2(H(m+3)-1)}{(2+m)^2 (3+m)}, \\
\xi_4(m) &=& -\frac{(2+m) (17+5 m)-2(3+m) (4+m) H(m+2)}{(m+1)(m+2)(m+3)}.
\eea
This completes the case for the event $\mc{E}$ (i.e. $g+w_{k-1} < 1$). Now when the event $\bar{\mc{E}}$ holds,
\bea
\PP(Z_*^l \le z | g,a,m, \bar{\mc{C}_1})_{\bar{\mc{E}}} &=& z (1-g+g z)^m -\frac{(1-2 g+(1-g) m) z^{2+m}}{(1-g)^2 (1+m) (2+m)} \n \\
&&+\frac{((1-g) (1+m)-g z) (1-g+g z)^{1+m}}{(1-g)^2 (1+m) (2+m)}.
\eea
Continuing as we did with the case $\mc{E}$,
\bea
\PP(Z_*^l \le z | m, \bar{\mc{C}_1})_{\bar{\mc{E}}} &:=& \int_0^1 \int_{1-g}^1 \PP(Z_*^l \le z | g,a,m, \bar{\mc{C}_1})_{\bar{\mc{E}}} f_A(a) f_G(g) \dd a \dd g \n \\
&=& \int_0^1 g \PP(Z_*^l \le z | g,a,m, \bar{\mc{C}_1})_{\bar{\mc{E}}} f_G(g) \dd g.
\label{secondkpint}
\eea
where we have used the fact that the expression $\PP(Z_*^l \le z | g,a,m, \bar{\mc{C}_1})_{\bar{\mc{E}}}$ is not a function of $a$. Evaluation of this integral is given in Section \ref{secondkpintsec}; the expression is
\bea
\PP(Z_*^l \le z | m, \bar{\mc{C}_1})_{\bar{\mc{E}}} &=& -\frac{2 z \left(1+m-3 z-m z+z^{2+m} (3+m-(1+m) z)\right)}{(1+m) (2+m) (3+m) (-1+z)^3} + \frac{-2z}{(m+1)(m+2)} \sum_{j=1}^{m+1} \frac{z^{m+1-j}}{j} \n \\
&& + \frac{2z}{(m+1)(m+2)^2(1-z)} + \frac{2(1+m+z)}{(m+1)(m+2)^2(m+3)(1-z)^2}- \frac{(6+2m)z^{m+2}}{(m+1)(m+2)} \n \\
&& + \frac{z^{m+2}}{m+1} + \frac{2H(m+1)z^{m+2}}{(m+1)(m+2)} - \frac{2(1+m+2z)z^{m+2}}{(m+1)(m+2)^2(1-z)} \n \\
&& - \frac{2(1+m+z)z^{m+3}}{(m+1)(m+2)^2(m+3)(1-z)^2}.
\eea
We again calculate the following term for the expected value
~\\
\bea
\bar{\E}[Z_*^l|m,\bar{\mc{C}_1}]_{\bar{\mc{E}}} &:=& \int_0^1 \PP(Z_*^l \le z | m) \dd z \n \\
&=& \frac{20+10 m+m^2-2 (3+m) H(1+m)}{(2+m) (3+m)^2} + \sum_{j=1}^{m+1} \frac{2}{j (-3+j-m) (1+m)(2+m)}. \n \\
\eea
Bringing together both cases $\mc{E}$ and $\bar{\mc{E}}$, we have
\bea
\E[Z_*^l|m, \bar{C}_1] &=& \int_0^1 (1- \PP(Z_*^l \le z | m,\bar{\mc{C}_1})) \dd z\n \\
&=& 1- \int_0^1 \PP(Z_*^l \le z | m,\bar{\mc{C}_1}) \dd z \n \\ 
&=& 1- \bar{\E}[Z_*^l|m, \bar{\mc{C}_1}]_{\mc{E}} - \bar{\E}[Z_*^l|m,\bar{\mc{C}_1}]_{\bar{\mc{E}}} \n \\
&=& 1 + \frac{1}{(m+1)(m+2)^3(m+3)^2} \left ( (186+472m + 448m^2 + 203m^3 + 45 m^4 + 4m^5) \right. \n \\
&& + (-244 - 454m -334m^2 - 124m^3 - 24m^4-2m^5) H(m+1) \n \\
&& \left . + (-48 - 88m - 60m^2 - 18m^3 - 2m^4) (H(m+1))^2 \right ) \n \\
&& + \sum_{j=1}^{m+1} \frac{2 \left(-4+j-4 m+j m-m^2-\left(j+(2+m)^2\right) H(j)+\left(j+(2+m)^2\right)H(3+m)\right)}{j (-3+j-m) (-2+j-m) (1+m) (2+m)}. \n \\
\eea
If the first item is critical,
\bea
\PP(Z_*\le z |g,m, \mc{C}_1) = (1-g+gz)^m.
\eea
Marginalizing over $G$ and taking the expectation gives
\bea
\PP(Z_* \le z |m,\mc{C}_1) &=& \int_0^1 \PP(Z_*^l \le z | g,m,\mc{C}_1) f_G(g) \dd g \n \\
&=& \int_0^1 (1-g+gz)^m (2-2g) \dd g \n \\
&=& \frac{2 \left(1+m-2 z-m z+z^{2+m}\right)}{(1+m) (2+m) (-1+z)^2}.
\eea
\bea
\E(Z_*|m, \mc{C}_1) &=& 1-\int_0^1 \PP(Z_* \le z |m, \mc{C}_1) \dd z \n \\
&=& 1 + \frac{2}{2+m} - \frac{2 H(m+1)}{m+1}.
\eea
Since the event $\mc{C}_1$ indicates $M = n-1$,
\bea
\E(Z_*| \mc{C}_1) &=& 1 + \frac{2}{n+1} - \frac{2 H(n)}{n}.
\eea
Finally, accounting for the distribution of $M$ with Lemma \ref{lemmakunif} gives the expression in the theorem:
\bea
&& \E \left [Z_*(n) \left \vert \sum_{i=1}^nW_i > B \right. \right ] \ge 1 + \frac{2}{n(n+1)} - \frac{2H(n)}{n^2} \n \\
&&+\frac{1}{n}\sum_{m=0}^{n-2} \left ( \sum_{j=1}^{m+1} T(j,m) + \left( (186+472m + 448m^2 + 203m^3 + 45 m^4 + 4m^5) \right. \right. \n \\
&& -(244 + 454m + 334m^2 + 124m^3 + 24m^4 + 2m^5) H(m+1) \n \\
&& \left. \left . -(48 + 88m + 60m^2 + 18m^3 + 2m^4) (H(m+1))^2 \right ) \frac{1}{(m+1)(m+2)^3(m+3)^2}  \right ), \n \\
\label{lotsofterms}
\eea
where
\bea
T(j,m) &:=& \frac{2\left(-4+j-4 m+j m-m^2-\left(j+(2+m)^2\right) H(j)\right)}{j (-3+j-m) (-2+j-m) (1+m) (2+m)}\n \\
&&+\frac{2(j+(2+m)^2)H(3+m)}{j (-3+j-m) (-2+j-m) (1+m) (2+m)}.
\eea
 \qed
~\\

We observe that the nested summation term may be omitted without significant loss in the performance bound. This is accomplished by showing that the argument of the sum is always positive.
\begin{lma}
For all $m>0$ and $1 \le j \le m+1$,
\bea
 \frac{\left(-4+j-4 m+j m-m^2-\left(j+(2+m)^2\right) H(j)+\left(j+(2+m)^2\right)H(3+m)\right)}{j (-3+j-m) (-2+j-m) (1+m) (2+m)} > 0.
 \eea
\end{lma}
\begin{proof}
The denominator is always positive, so we focus on the numerator. There numerator consists of two parts,
\bea
N_1(j,m) &:=& (4-j)(m+1)+m^2,  \\
N_2(j,m) &:=& (j+(2+m)^2)\sum_{i=j+1}^{m+3} \frac{1}{i} .
\eea
Our goal is to show that $N_2(j,m) > N_1(j,m)$ always holds. The difference equation for $N_2(j,m)$ with respect to $j$ satisfies
\bea
\Delta(N_2(j,m)) &:=& N_2(j+1,m) - N_2(j,m) \n \\
&=& \sum_{i=j+2}^{m+3} \frac{1}{i} + (j+(2+m)^2)\sum_{i=j+2}^{m+3} \frac{1}{i} - (j+(2+m)^2)\sum_{i=j+1}^{m+3} \frac{1}{i} \n \\
&=& \sum_{i=j+2}^{m+3}\frac{1}{i} -\frac{j+(2+m)^2}{j+1} \n \\
&\le& \frac{m-j+2}{j+2}-\frac{j+(2+m)^2}{j+1} \n \\
&<& \frac{m-j+2}{j+1}-\frac{j+(2+m)^2}{j+1} \n \\
&=& \frac{-2-3m-m^2-2j}{j+1} \n \\
&\le& \frac{-4-3m-m^2}{m+2}.
\eea
For the other term, we have
\bea
\Delta(N_1(j,m)) = -(m+1).
\eea
Both $N_1(j,m)$ and $N_2(j,m)$ are decreasing in $j$ and $N_2(j,m)$ decreases at a greater rate. We approximate $N_2(j,m)$ with the following:
\bea
H(m+3)-H(j) =  \sum_{i=j+1}^{m+3} \frac{1}{i} \ge \int_{j+1}^{m+4} \frac{1}{x}\dd x = \log \left (\frac{m+4}{j+1} \right ).
\eea
Looking at $j=1$,
\bea
N_1(1,m) &=& 3 +3m+m^2,  \\
N_2(1,m) &=& (5+4m+m^2) \log \left (\frac{m+4}{2} \right),
\eea
guaranteeing $N_2(1,m) > N_1(1,m)$. With consideration of starting points and slopes for the two numerator terms, ensuring that $N_2(m+1,m)>N_1(m+1,m)$ is sufficient for the lemma. We have
\bea
N_1(m+1,m) &=& 3+2m, \\
N_2(m+1,m) &=& (m+1+(2+m)^2)\left (\frac{1}{m+2}+\frac{1}{m+3} \right ) \n \\
&>& (5+5m+m^2)\left ( \frac{2}{m+3} \right ) \n \\
&=& \frac{10+10m+2m^2}{m+3} \n \\
&>& 3+2m.
\eea
\end{proof}

\nin \textit{Proof of Theorem \ref{kpasym}}. We will show that $\lim_{n\rightarrow \infty} \E[Z_*|\cdot] = 1$, so we are interested in bounding the rate at which $1-\E[Z_*|\cdot]$ approaches $0$. Accordingly, we are only concerned with the negative terms in \eqref{lotsofterms}. The magnitudes of these terms are
\bea
T_1(n) &=& \frac{2H(n)}{n^2}, \\
T_2(n) &=& \frac{1}{n} \sum_{m=0}^{n-2} \frac{(244+454m+334m^2+24m^4+2m^5)H(m+1)}{(m+1)(m+2)^3(m+3)^2}, \\
T_3(n) &=& \frac{1}{n} \sum_{m=0}^{n-2}\frac{(48+88m+60m^2+18m^3+2m^4)(H(m+1))^2}{(m+1)(m+2)^3(m+3)^2}.
\eea
The second and third terms are decreasing in $m$, so they are bounded by their respective integrals. Using a logarithmic bound on the harmonic numbers, we have
\bea
T_1(n) &=& O \left ( \frac{\log{n}}{n^2} \right ), \\
T_2(n) &=& \frac{1}{n} \sum_{m=0}^{n-2} O \left (\frac{\log m }{m} \right ) \n \\
&=& \frac{1}{n} \int_0^{n-1} O \left (\frac{\log m}{m} \right ) \dd m \n \\
&=& O \left( \frac{\log^2n}{n} \right ), \\
T_3(n) &=& \frac{1}{n} \sum_{m=0}^{n-2} O \left (\frac{\log^2 m }{m^2} \right ) \n \\
&=& \frac{1}{n} \int_0^{n-1} O \left (\frac{\log^2 m }{m^2} \right ) \dd m \n \\
&=& O \left ( \frac{ \log^2 n}{n^2} \right ).
\eea
The largest growth rate is $O(\frac{\log^2n}{n}$). Also, we have that $\lim_{n \rightarrow \infty} \E[Z_*(n) | \cdot] = 1$ since the gain has a natural upper bound of unit value. \qed
~\\
~\\

\appendix
\section{Appendix}
\label{apdx}
The following lemma is used in integral evaluations described in this section.
\begin{lma}
For constant values $\kappa_1,~\kappa_2$ and nonnegative integer $\theta$,
\bea
\int \frac{(\kappa_1+\kappa_2x)^\theta}{x} \dd x = \kappa_1^\theta \log(x) + \sum_{j=1}^\theta \frac{\kappa_1^{\theta-j}(\kappa_1+\kappa_2x)^j}{j}.
\eea
\label{theintegral}
\end{lma}
\begin{proof}
We begin by noting that 
\bea
\int \frac{(\kappa_1+\kappa_2x)^\theta}{x} \dd x &=& \int \kappa_2 (\kappa_1+\kappa_2x)^{\theta-1} \dd x + \int \frac{\kappa_1(\kappa_1+\kappa_2x)^{\theta-1}}{x}\dd x  \n \\
&=& \frac{(\kappa_1+\kappa_2x)^\theta}{\theta} + \kappa_1 \int \frac{(\kappa_1+\kappa_2x)^{\theta-1}}{x}\dd x.
\eea
The statement of the lemma clearly holds for $\theta = 0$. Assuming that it holds for $\theta = t$, we have for $\theta = t+1$,
\bea
\int \frac{(\kappa_1+\kappa_2x)^{t+1}}{x} \dd x &=& \frac{(\kappa_1+\kappa_2x)^{t+1}}{t+1} + \kappa_1 \left ( \kappa_1^t \log(x) + \sum_{j=1}^t \frac{\kappa_1^{t-j}(\kappa_1+\kappa_2x)^j}{j}\right ) \n \\
&=& \kappa_1^{t+1} \log(x) + \sum_{j=1}^{t+1} \frac{\kappa_1^{t+1-j}(\kappa_1+\kappa_2x)^j}{j}.
\eea
The property then holds for all $\theta$ by induction.

\end{proof}

\subsection{Integral evaluation of \eqref{theintupper}}
\label{theintsec}
To simplify expressions, we use $A := W_{K-1}$. Also recall that $M:= n-K$. The integral is
\bea
\PP(V_*^u>v|m,\bar{\mc{C}}_1) &=& \int_0^1 \int_0^1 \PP(G>v|a,g,\bar{\mc{C}}_1) \PP(\tilde{V}^u>v|a,g,\bar{\mc{C}}_1) \left(\PP(V^u>v|a,g,\bar{\mc{C}}_1)\right)^m f_A(a) f_G(g) \dd a \dd g. \n \\
\eea
This may be evaluated by considering regions where the arguments have simple analytical descriptions as a function of $a$ and $g$. We begin by noting that $\PP(G>v|a,g,\bar{\mc{C}}_1) = \I(v<g)$, so we may restrict our analysis to regions where $v < g$. For the integral evaluation of $(a,g) \in R_j$, we use the notation
\be
\rho_j(v,m) = \iint_{R_j}  \PP(G>v|a,g,\bar{\mc{C}}_1) \PP(\tilde{V}^u>v|a,g,\bar{\mc{C}}_1) \left(\PP(V^u>v|a,g,\bar{\mc{C}}_1)\right)^m f_A(a) f_G(g) \dd a \dd g.
\ee
The relevant regions are shown in Figure \ref{regions}, where different enumerations are necessary for $v \le \frac{1}{2}$ and $v > \frac{1}{2}$. The values of $(\PP(V^u>v|a,g,\bar{\mc{C}}_1))^m$ and $\PP(\tilde{V}^u>v|a,g,\bar{\mc{C}}_1)$ are shown in Table \ref{int_table}. Note that in many cases, the $\frac{1}{1-g}$ factor from $\PP(G > v|a,g,\bar{\mc{C}}_1)$ cancels with the $(1-g)$ factor from $f_G(g)$, which simplifies the expression.
\begin{table}[h]
\begin{center}
\caption{Arguments of (\ref{theintupper}) for regions shown in Figure \ref{regions}. \label{int_table}}
\begin{tabular}{|c|c|c|}
\hline
Region &  $(\PP(V^u>v|a,g,\bar{\mc{C}}_1))^m$ & $\PP(\tilde{V}^u>v|a,g,\bar{\mc{C}}_1) $ \\
\hline
$R_1$ & $(1-v-a)^m$ & $(1-g-a)/(1-g)$ \\
$R_2$ & $(g-v)^m$ & 0 \\
$R_3$ & $(1-2v)^m$ & $(1-g-v)/(1-g)$ \\
$R_4$ & $(a+g-2v)^m$ & $(a-v)/(1-g)$ \\
$R_5$ & $(a+g-2v)^m$ & $(a-v)/(1-g)$ \\
$R_6$ & $(1-v)^m$ & 1 \\
\hline
$R_7$ & $(1-v-a)^m$ & $(1-g-a)/(1-g)$ \\
$R_8$ & $(g-v)^m$ & 0 \\
$R_9$ & $(a+g-2v)^m$ & $(a-v)/(1-g)$ \\
$R_{10}$ & $(1-v)^m$ & 1 \\
\hline
\end{tabular}
\end{center}
\end{table}

\begin{figure}[h]
\begin{center}
\includegraphics[scale=0.8]{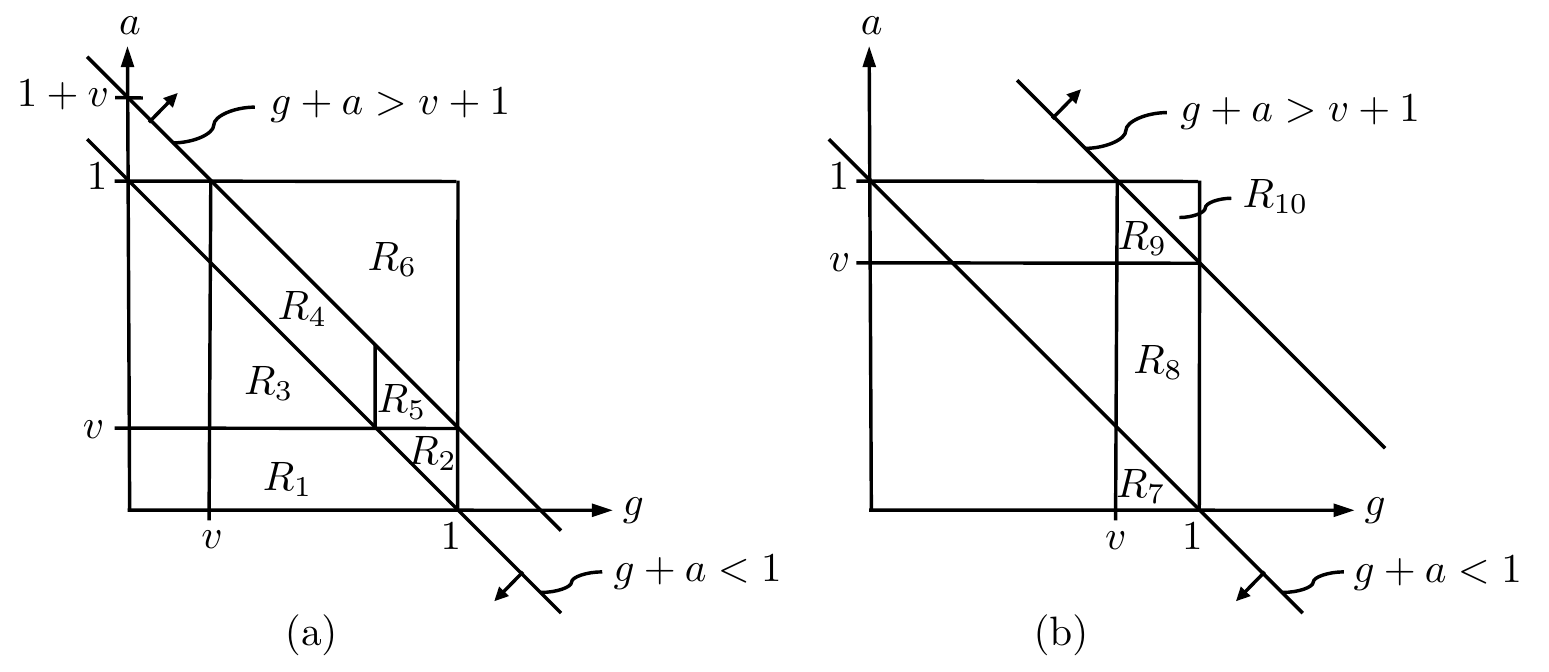}
\caption{Integration regions for (a) $v \le \frac{1}{2}$ and (b) $v > \frac{1}{2}$.\label{regions}}
\end{center}
\end{figure}
Regions 1-6 correspond to the case where $v \le \frac{1}{2}$.
\bea
\rho_1(v,m) &=& \int _0^v\int _v^{1-a}2(1-v-a)^m(1-g-a)\dd g \dd a = \int_0^v (1-a-v)^{2+m} \dd a  \n \\
&=& \frac{-(1-2 v)^{3+m}+(1-v)^{3+m}}{3+m}. \\
~ \n \\
\rho_2(v,m) &=&0. \\
~ \n \\
\rho_3(v,m) &=& \int _v^{1-v}\int _v^{1-a}2(1-2 v)^m (1-g-v)\dd g \dd a = \int_v^{1-v} (3v-a-1) (1-2 v)^m (v+a-1) \dd a \n \\
&=&  \frac{2}{3} (1-2 v)^{3+m}.
\eea
\bea
\rho_4(v,m) &=& \int _v^{1-v}\int _{1-g}^{v+1-g}2(a+g-2 v)^m (a-v) \dd a \dd g \n \\
&=&\frac{1}{(1+m)(2+m)} \int_v^{1-v} \left ( -2 (1-2 v)^{1+m}+4 g (1-2 v)^{1+m}-2 m (1-2 v)^{1+m}+2 g m (1-2 v)^{1+m} \right. \n \\
&& +2 (1-v)^{1+m}-4 g (1-v)^{1+m}+2 m (1-v)^{1+m}-2 g m (1-v)^{1+m}+2 m (1-2 v)^{1+m} v \n \\
&& +\left. 2 (1-v)^{1+m} v \right ) \dd g \n \\
&=& \frac{1}{(1+m)(2+m)} \left ( m (1-2 v)^{3+m}+m (1-v)^m+2 (1-v)^m v-3 m (1-v)^m v \right. \n \\ 
&&\left. -6 (1-v)^m v^2+2 m (1-v)^m v^2+4 (1-v)^m v^3 \right ).
\eea
\bea
\rho_5(v,m) &=& \int _{1-v}^1\int _v^{1+v-g}2(a+g-2 v)^m (a-v) \dd a \dd g \n \\
&=& \frac{1}{(1+m)(2+m)}\int_{1-v}^1 \left ( 2 (1-v)^{1+m}- 4 g (1-v)^{1+m}+2 m (1-v)^{1+m}-2 g m (1-v)^{1+m} \right. \n \\
&& \left. +2 (g-v)^{2+m}+2 (1-v)^{1+m} v \right ) \dd g \n \\
&=& \frac{1}{(1+m)(2+m)(3+m)} \left ( -2 (1-2 v)^{3+m}+2 (1-v)^{1+m}-10 (1-v)^{1+m} v -2 m (1-v)^{1+m} v \right .\n \\ 
&& \left. +14 (1-v)^{1+m} v^2+7 m (1-v)^{1+m} v^2 +m^2 (1-v)^{1+m} v^2 \right ). \\
~ \n \\
\rho_6(v,m) &=& \int _v^1\int _{1+v-g}^1(1-v)^m (2-2 g) \dd a \dd g = \int_v^1 (2-2 g) (1-v)^m (g-v) \dd g = \frac{1}{3} (1-v)^{3+m}.
\eea
Summing all terms of $\PP(V_*^u>v|m, \bar{\mc{C}}_1)$ for $v \le \frac{1}{2}$ gives
\bea
\PP(V_*^u>v|m,\bar{\mc{C}}_1)_{\le \frac{1}{2}} &:=& \rho_1(v,m)+\rho_2(v,m)+\rho_3(v,m)+\rho_4(v,m)+\rho_5(v,m)+\rho_6(v,m) \n \\
&=& \frac{1}{3(3+m)} \left ( 2m(1-2 v)^m+m (1-v)^m+ 9(1-v)^{3+m}-12m(1-2 v)^m v \right. \n \\ 
&& -3m (1-v)^m v + 24m(1-2 v)^m v^2+3m(1-v)^m v^2-16 m (1-2 v)^m v^3 \n \\
&& \left. -m (1-v)^m v^3 \right ).
\eea
Regions 7-10 are for the case $v>\frac{1}{2}$.
\bea
\rho_7(v,m) &=& \int _v^1\int _0^{1-g}2(1-v-a)^m (1-g-a) \dd a \dd g \n \\
&=& \frac{1}{(1+m)(2+m)} \int_v^1 \left ( 2 (1-v)^{1+m}-4 g (1-v)^{1+m}+2 m (1-v)^{1+m}-2 g m (1-v)^{1+m}\right. \n \\
&& \left.+2 (g-v)^{2+m} +2 (1-v)^{1+m} v\right )\dd g = \frac{(1-v)^{3+m}}{3+m}. \\
~ \n \\
\rho_8(v,m) &=& 0. \\
~ \n \\
\rho_9(v,m) &=& \int _v^1\int _v^{1+v-g}2(a+g-2 v)^m(a-v) \dd a \dd g \n \\
&=& \frac{1}{(1+m)(2+m)} \int_v^1 \left (2 (1-v)^{1+m}-4 g (1-v)^{1+m}+2 m (1-v)^{1+m}-2 g m (1-v)^{1+m}\right. \n \\
&&+2 (g-v)^{2+m}+\left. 2 (1-v)^{1+m} v\right ) \dd g = \frac{(1-v)^{3+m}}{3+m}. \\
~ \n \\
\rho_{10}(v,m) &=& \int _v^1\int _{1+v-g}^1(1-v)^m (2-2 g) \dd a \dd g = \int_v^1 (2-2 g) (1-v)^m (g-v) \dd g = \frac{1}{3} (1-v)^{3+m}.
\eea
Summing these terms yields for $v > \frac{1}{2}$,
\bea
 \PP(V_*^u>v|m,\bar{\mc{C}}_1)_{>\frac{1}{2}}&:=& \rho_7(v,m) + \rho_8(v,m)+\rho_9(v,m)+\rho_{10}(v,m) \n \\
&=& \frac{1}{3} (1-v)^{3+m}+\frac{2 (1-v)^{3+m}}{3+m}.
\eea
In summary, we have
\bea
\PP(V_*^u>v|m, \bar{\mc{C}}_1) =
\left \{ \begin{array}{cc}
\PP(V_*^u>v|m,\bar{\mc{C}}_1)_{\le \frac{1}{2}} & v \le \frac{1}{2} \\
\PP(V_*^u>v|m,\bar{\mc{C}}_1)_{>\frac{1}{2}} & v > \frac{1}{2}.
\end{array} \right.
\eea

\subsection{Integral evaluation of \eqref{firstkpint}}
\label{firstkpintsec}
We wish to evaluate
\bea
\PP(Z_*^l \le z | m, \bar{\mc{C}_1})_{\mc{E}} &:=& \int_0^1 \int_0^{1-g} \PP(Z_*^l \le z |g,a,m, \bar{\mc{C}_1})_{\mc{E}} f_A(a) f_G(g) \dd a \dd g,
\eea
where
\bea
\PP(Z_*^l \le z |g,a,m, \bar{\mc{C}_1})_{\mc{E}} &=& z(1-g+gz)^m + \frac{(1-g+gz)^{m+1}(1-g+m-gm-gz)}{(1-g)a(m+1)(m+2)} \n \\
&& -\frac{(1-g+gz + a(1-z))^{m+1} \left (1-g+m-gm-gz + a(-1-m+z+mz) \right )}{(1-g)a(m+1)(m+2)}. \n \\
\eea
We first determine the following using the fact that $A$ follows distribution $\mc{U}[0,1]$
\bea
\int \PP(Z_*^l \le z |g,a,m,\bar{\mc{C}_1})_{\mc{E}} f_A(a) \dd a = \int \PP(Z_*^l \le z |g,a,m, \bar{\mc{C}_1})_{\mc{E}} \dd a.
\eea
The following constants simplify the expression:
\bea
\lambda_1 &:=& 1-g+gz, \\
\lambda_2 &:=& z-1, \\
\lambda_3 &:=& 1-g+m-gm-gz, \\
\lambda_4 &:=& -1-m+z+mz, \\
\lambda_5 &:=& \frac{-1}{(1-g)(m+1)(m+2)}.
\eea
This gives
\bea
\int \PP(Z_*^l \le z |g,a,m,\bar{\mc{C}_1})_{\mc{E}} \dd a &=& \int \left (z \lambda_1^m - \frac{\lambda_5  \lambda_3 \lambda_1^{m+1}}{a} + \lambda_5 \lambda_3 \frac{(\lambda_1+a\lambda_2)^{m+1}}{a} + \lambda_5 \lambda_4(\lambda_1+a\lambda_2)^{m+1}\right ) \dd a \n \\
&=& az\lam_1^m - \lam_5 \lam_3 \lam_1^{m+1} \log(a) + \lam_5 \lam_3 \left (\lam_1^{m+1} \log(a) + \sum_{j=1}^{m+1} \frac{\lam_1^{m+1-j}(\lam_1+\lam_2 a)^j}{j} \right ) \n \\
&& + \frac{\lam_5 \lam_4}{\lam_2(m+2)} (\lam_1 + \lam_2 a)^{m+2} \n \\
&=& az\lam_1^m + \lam_5 \lam_3 \sum_{j=1}^{m+1} \frac{\lam_1^{m+1-j}(\lam_1+\lam_2 a)^j}{j} + \frac{\lam_5 \lam_4}{\lam_2(m+2)} (\lam_1 + \lam_2 a)^{m+2},
\eea
where we have made use of the integral identity from Lemma \ref{theintegral}. Evaluating over the domain of integration gives
\bea
\int_0^{1-g} \PP(Z_*^l \le z |g,a,m, \bar{\mc{C}_1})_{\mc{E}} \dd a &=& (1-g)z\lam_1^m + \lam_5 \lam_3 \left ( \sum_{j=1}^{m+1} \frac{\lam_1^{m+1-j}z^j}{j} - \lam_1^{m+1} H(m+1) \right) \n \\
&& + \frac{\lam_5 \lam_4 (z^{m+2} -\lam_1^{m+2})}{\lam_2(m+2)}.
\eea
Next, we calculate
\bea
\int_0^1 \int_0^{1-g} \PP(Z_*^l \le z |g,a,m, \bar{\mc{C}_1})_{\mc{E}} f_A(a) f_G(g) \dd a \dd g &=& \int_0^1 \left ( \int_0^{1-g} \PP(Z_*^l \le z |g,a,m, \bar{\mc{C}_1})_{\mc{E}} \dd a \right ) (2-2g) \dd g. \n \\
\eea
We integrate each additive term separately:
\bea
\rho_1(m,z) &=& \int_0^1 (1-g)z\lambda_1^m(2-2g) \dd g \n \\
&=& \int_0^1 2(1-g)^2z(1-g+gz)^m \dd g \n \\
&=& -\frac{2 z \left(2+m^2 (-1+z)^2+m (-1+z) (-3+5 z)-2 z \left(3-3 z+z^{2+m}\right)\right)}{(1+m) (2+m) (3+m) (-1+z)^3}.
\eea
\bea
\rho_{2j}(m,z) &=& \int_0^1 \lam_5 \lam_3 \frac{\lam_1^{m+1-j}z^j}{j}(2-2g) \dd g \n \\&=& \int_0^1-\frac{2(1-g+m-gm-gz)(1-g+gz)^{m+1-j}z^j}{(m+1)(m+2)j} \dd g\n \\
&=& \frac{2z^{3+m} (j+(2+m) (-2+z)-j z)}{j (-3+j-m) (-2+j-m) (1+m) (2+m) (-1+z)^2} \n \\
&& + \frac{2z^j (-j (1+m) (-1+z)+(2+m) (-1+m (-1+z)+2 z))}{j (-3+j-m) (-2+j-m) (1+m) (2+m) (-1+z)^2}.
\eea
\bea
\rho_3(m,z) &=& \int_0^1 -\lam_5 \lam_3 \lam_1^{m+1} H(m+1) (2-2g) \dd g \n \\
&=& \int_0^1 \frac{2H(m+1)(1-g+m-gm-gz)(1-g+gz)^{m+1}}{(m+1)(m+2)} \dd g \n \\
&=& -\frac{2 H(m+1) \left(-1+m (-1+z)+2 z+(-2+z) z^{3+m}\right)}{(1+m) (2+m) (3+m) (-1+z)^2}.
\eea
\bea
\rho_4(m,z) &=& \int_0^1 \frac{\lam_5 \lam_4 (z^{m+2}-\lam_1^{m+2})}{\lam_2 (m+2)}(2-2g) \dd g \n \\
&=& -\frac{2(-1-m+z+mz)(z^{m+2} - (1-g+gz)^{m+2})}{(m+1)(m+2)^2(-1+z)} \dd g \n \\
&=& -\frac{2}{(2+m)^2 (3+m) (-1+z)}-\frac{2 z^{2+m}}{(2+m)^2}+\frac{2 z^{3+m}}{(2+m)^2 (3+m) (-1+z)}.
\eea
With these terms, we have
\bea
\PP(Z_*^l \le z | m,\mc{E}, \bar{\mc{C}_1}) &=& \rho_1(m,z) + \sum_{j=1}^{m+1} \rho_{2j}(m,z) + \rho_3(m,z) + \rho_4(m,z).
\eea

\subsection{Integral evaluation of \eqref{secondkpint}}
\label{secondkpintsec}
The integral is 
\bea
\PP(Z_*^l \le z|m, \bar{\mc{C}_1})_{\bar{\mc{E}}} &=& \int_0^1 g \PP(Z_*^l \le z | g,a,m, \bar{\mc{C}_1})_{\bar{\mc{E}}} f_G(g) \dd g,
\eea
where
\bea
\PP(Z_*^l \le z | g,a,m, \bar{\mc{C}_1})_{\bar{\mc{E}}} &=& z (1-g+g z)^m -\frac{(1-2 g+(1-g) m) z^{2+m}}{(1-g)^2 (1+m) (2+m)} \n \\
&&+\frac{((1-g) (1+m)-g z) (1-g+g z)^{1+m}}{(1-g)^2 (1+m) (2+m)}.
\eea
For the first term in $\PP(Z_*^l \le z | m, \bar{\mc{C}_1})_{\bar{\mc{E}}}$,
\bea
\int_0^1 gz(1-g+gz)^m(2-2g) \dd g = -\frac{2 z \left(1+m-3 z-m z+z^{2+m} (3+m-(1+m) z)\right)}{(1+m) (2+m) (3+m) (-1+z)^3}.
\eea
To find the indefinite integral of the second term in $\PP(Z_*^l \le z | m, \bar{\mc{C}_1})_{\bar{\mc{E}}}$, we use the substitution $g = 1-e.$
\bea
 && \int -\frac{2g(1-2 g+(1-g) m) z^{2+m}}{(1-g) (1+m) (2+m)} \dd g \n \\
&=& \int \frac{2(1-e)(1-2(1-e)+em) z^{m+2}}{e(m+1)(m+2)} \dd e \n \\
&=& \int \frac{(-2 +2e(m+3) -2e^2(m+2))z^{m+2}}{e(m+1)(m+2)} \dd e \n \\
&=& \int \frac{-2z^{m+2}}{e(m+1)(m+2)} \dd e + \int \frac{2(m+3)z^{m+2}}{(m+1)(m+2)} \dd e + \int \frac{-2ez^{m+2}}{(m+1)} \dd e \n \\
&=& \frac{-2z^{m+2}}{(m+1)(m+2)} \log(e) + \frac{2e(3+m)z^{m+2}}{(m+1)(m+2)} - \frac{e^2z^{m+2}}{(m+1)}.
\eea
For the indefinite integral of the final term in $\PP(Z_*^l \le z | m, \bar{\mc{C}_1})_{\bar{\mc{E}}}$, we again use the substitution $g = 1-e$.
\bea
&&\int \frac{2g((1-g) (1+m)-g z) (1-g+g z)^{1+m}}{(1-g)(1+m) (2+m)} \dd g \n \\
&=& \int \frac{-2(1-e)(e(m+1)-(1-e)z)(1+(1-e)(z-1))^{m+1}}{e(m+1)(m+2)} \dd e\n \\
&=& \int \frac{(2z - 2e(1+m+2z) - 2e^2(-1-m-z))(z+e(1-z))^{m+1}}{e(m+1)(m+2)} \dd e\n \\
&=& \int \frac{2z(z+e(1-z))^{m+1}}{e(m+1)(m+2)} \dd e + \int  \frac{-2(1+m+2z)(z+e(1-z))^{m+1}}{(m+1)(m+2)} \dd e \n \\
&& + \int \frac{-2e(-1-m-z)(z+e(1-z))^{m+1}}{(m+1)(m+2)} \dd e \n \\
&=& \frac{2z}{(m+1)(m+2)} \left( z^{m+1} \log(e) + \sum_{j=1}^{m+1} \frac{z^{m+1-j}(z+e(1-z))^j}{j} \right ) \n \\
&& - \frac{2(1+m+2z)(z+e(1-z))^{m+2}}{(m+1)(m+2)^2(1-z)} - \frac{2e(-1-m-z)(z+e(1-z))^{m+2}}{(m+1)(m+2)^2(1-z)} \n \\
&& + \frac{2(-1-m-z)(z+e(1-z))^{m+3}}{(m+1)(m+2)^2(m+3)(1-z)^2} .
\eea
Note that we have used the integral identity from Lemma \ref{theintegral}. For the second and third terms, we have
\bea
 &&\int_0^1 \left (\frac{((1-g) (1+m)-g z) (1-g+g z)^{1+m}}{(1-g)^2 (1+m) (2+m)} -\frac{(1-2 g+(1-g) m) z^{2+m}}{(1-g)^2 (1+m) (2+m)} \right )\dd g \n \\
 &=& \frac{2z}{(m+1)(m+2)}\sum_{j=1}^{m+1} \frac{z^{m+1-j}(z+e(1-z))^j}{j} - \frac{2(1+m+2z)(z+e(1-z))^{m+2}}{(m+1)(m+2)^2(1-z)} \n \\
&& -  \frac{2e(-1-m-z)(z+e(1-z))^{m+2}}{(m+1)(m+2)^2(1-z)} + \frac{2(-1-m-z)(z+e(1-z))^{m+3}}{(m+1)(m+2)^2(m+3)(1-z)^2} \n \\
&&\left. + \frac{2e(3+m)z^{m+2}}{(m+1)(m+2)} - \frac{e^2z^{m+2}}{(m+1)} \right \vert^{e=0}_{e=1} \n \\
&=& \frac{-2z}{(m+1)(m+2)} \sum_{j=1}^{m+1} \frac{z^{m+1-j}}{j} + \frac{2z}{(m+1)(m+2)^2(1-z)} + \frac{2(1+m+z)}{(m+1)(m+2)^2(m+3)(1-z)^2} \n \\
&&- \frac{(6+2m)z^{m+2}}{(m+1)(m+2)} + \frac{z^{m+2}}{m+1} + \frac{2H(m+1)z^{m+2}}{(m+1)(m+2)} - \frac{2(1+m+2z)z^{m+2}}{(m+1)(m+2)^2(1-z)} \n \\
&& - \frac{2(1+m+z)z^{m+3}}{(m+1)(m+2)^2(m+3)(1-z)^2}.
\eea
Altogether,
\bea
\PP(Z_*^l \le z | m,\bar{\mc{C}_1})_{\bar{\mc{E}}} &=& -\frac{2 z \left(1+m-3 z-m z+z^{2+m} (3+m-(1+m) z)\right)}{(1+m) (2+m) (3+m) (-1+z)^3} + \frac{-2z}{(m+1)(m+2)} \sum_{j=1}^{m+1} \frac{z^{m+1-j}}{j} \n \\
&& + \frac{2z}{(m+1)(m+2)^2(1-z)} + \frac{2(1+m+z)}{(m+1)(m+2)^2(m+3)(1-z)^2}- \frac{(6+2m)z^{m+2}}{(m+1)(m+2)} \n \\
&& + \frac{z^{m+2}}{m+1} + \frac{2H(m+1)z^{m+2}}{(m+1)(m+2)} - \frac{2(1+m+2z)z^{m+2}}{(m+1)(m+2)^2(1-z)} \n \\
&& - \frac{2(1+m+z)z^{m+3}}{(m+1)(m+2)^2(m+3)(1-z)^2}.
\eea